\documentclass{article}
\usepackage[utf8]{inputenc}
\usepackage{amsthm}
\usepackage{xcolor}
\usepackage{amsmath}
\usepackage{graphicx}
\usepackage{subcaption}
\usepackage[ruled]{algorithm2e} 

\usepackage{amsfonts} 
\usepackage{fullpage}
\usepackage{appendix}

\theoremstyle{definition}
\newtheorem{definition}{Definition}[section]
\theoremstyle{theorem}
\newtheorem{theorem}{Theorem}[section]
\theoremstyle{assumption}
\newtheorem{assumption}{Assumption}[section]
\theoremstyle{proposition}
\newtheorem{proposition}{Proposition}[section]
\theoremstyle{lemma}
\newtheorem{lemma}{Lemma}[section]
\newtheorem{corollary}{Corollary}[section]
\usepackage{hyperref}
\hypersetup{
    colorlinks=true,
    linkcolor=blue,
    filecolor=gray,      
    urlcolor=blue,
    citecolor=blue,
}
\DeclareMathOperator*{\argmax}{arg\,max}
\DeclareMathOperator*{\argmin}{arg\,min}

\title{Robustness of Online Proportional Response in Stochastic Online Fisher Markets: a Decentralized Approach}

\author{Yongge Yang\thanks{Economics and Management School, Wuhan University, Wuhan, Hubei, People's Republic of China. yonggeyang@whu.edu.cn} \and Yu-Ching Lee\thanks{Department of Industrial Engineering, National Tsing Hua University. 101 Sec 2, Kuang-Fu Rd, Hsinchu City 300, Taiwan. yclee@ie.nthu.edu.tw}\and Po-An Chen\thanks{Institute of Information Management, National Yang Ming Chiao Tung University. 1001 University Rd, Hsinchu City 300, Taiwan. poanchen@nycu.edu.tw} \and Chuang-Chieh Lin\thanks{Department of Computer Science and Information Engineering
Tamkang University. 151 Yingzhuan Rd, New Taipei City 251301, Taiwan}}

\begin{document}

\maketitle
\begin{abstract}
        \textbf{Problem Definition:} This study is focused on periodic Fisher markets where items with time-dependent and stochastic values are regularly replenished and buyers aim to maximize their utilities by spending budgets on these items. Traditional approaches of finding a market equilibrium in the single-period Fisher market rely on complete information about buyers' utility functions and budgets. However, it is impractical to consistently enforce buyers to disclose this private information in a periodic setting.
    \textbf{Methodology/results}: We introduce a distributed auction algorithm, \emph{online proportional response}, wherein buyers update bids solely based on the randomly fluctuating values of items in each period. The market then allocates items based on the bids provided by the buyers. We show connections between the online proportional response and the online mirror descent algorithm. 
    Utilizing the known Shmyrev convex program, a variant of the Eisenberg-Gale convex program that establishes market equilibrium of a Fisher market, two performance metrics are proposed: the fairness regret is the cumulative difference in the objective value of a stochastic Shmyrev convex program between an online algorithm and an offline optimum, and the individual buyer's regret 
    gauges the deviation in terms of utility for each buyer between the online algorithm and the offline optimum. Our algorithm attains a {problem-dependent upper bound in fairness regret} under stationary inputs. This bound is contingent on the number of items and buyers. Additionally, we conduct analysis of regret under various non-stationary stochastic input models to demonstrate the algorithm's efficiency across diverse scenarios.
    \textbf{Managerial implications}: The online proportional response algorithm addresses privacy concerns by allowing buyers to update bids without revealing sensitive information and ensures decentralized decision-making, fostering autonomy and potential improvements in buyer satisfaction. Furthermore, our algorithm is universally applicable to many worlds and shows the robust performance guarantees.
\end{abstract}
  \section{Introduction}
Market equilibrium, a central concept in economics, plays a crucial role in understanding how markets function and achieve efficiency. In a Fisher market, which is one of the classical mathematical market models, there are $N$ buyers who start with money they have no desire to keep and 1 supplier who has $M$ items to be sold in exchange for money. A market equilibrium is characterized by a set of prices and allocations of items, ensuring that each buyer can maximize their utility within their budget constraint while simultaneously achieving market clearance where demand matches supply. Notably, the existence of market equilibrium of the Arrow–Debreu market was established in \cite{arrow1954existence} under mild conditions on the utility function of buyers; \cite{eisenberg1959consensus} proposed a centralized convex program to maximize the Nash social welfare (the geometric mean of the buyers' utilities), where the corresponding solution is the market equilibrium of the Fisher market.

Most of the existing works on market equilibrium have predominantly focused on static markets, wherein the characteristics of items, supplies, buyers, valuations, and budgets remain unchanged over time. However, in practical scenarios, markets exhibit high dynamics, wherein these properties may actually vary over time. To be more precise, we consider a scenario where the buyers' valuations toward items change over time while the types of items and buyers remain constant. {We now list some examples to describe the randomness of buyers' valuations:}

\begin{enumerate}
\item Internet advertising: Multiple advertisers compete for multiple ad placements on the website or search engine. When a user visits a website or performs a search, the advertising platform needs to allocate the ad slots to the advertisers based on factors like bid amount, relevance, and user preferences. Generally, these factors determine the advertisers' valuations on the ad slots and the heterogeneous users. 

\item Cloud computing: Multiple users submit jobs that require computational resources. The platform needs to allocate the resources to the incoming jobs efficiently, considering factors like resource requirements, deadlines, and user priorities.

\item Stockpile allocation: A government agency maintains a stockpile of emergency supplies, such as medical equipment, food, and water, and it needs to dynamically assess the demand of each region and allocate the available supplies accordingly. The valuation of emergency supplies in each region can change due to factors such as population density, vulnerability to specific types of disasters, and random ongoing emergency situations. Each region provides the list of  emergency supplies they want the most, which can be seen as a valuation, and the government distributes the emergency supplies. 
\end{enumerate}

Although computing the market equilibrium via the Eisenberg and Gale (EG) convex program is polynomial-time solvable (\cite{jain2007polynomial}), the nature of centralized optimization requires the market to have complete information about the buyers' valuations of items and their budgets. However, the buyers can misreport valuations to obtain better allocations, see~\cite{adsul2010nash}, \cite{branzei2014fisher}, and~\cite{branzei2017nash}. It is impractical to enforce buyers to consistently report their true valuations, thereby ensuring adherence to the market equilibrium in each period. 

Motivated by practical considerations, this paper delves into the exploration of an online Fisher market setting. In this setting, items continually emerge for sale in each period, while the buyers' valuations of these items undergo random fluctuations based on factors specific to the current period, such as quality, quantity\footnote{The quantity is usually normalized to one, so it affects the valuation.}, market trends, and other relevant influences. Furthermore, we assume that the items are perishable, meaning that allocation decisions must be made immediately, or else the current item will be lost. 
For instance, in online advertising, multiple ad placements or URLs are available when a user visits a website, advertisers need to bid for and secure ad placements in real-time to reach their target audience effectively. The similar settings have been studied in \cite{manshadi2021fair}, \cite{calmon2021revenue}, \cite{sinclair2022sequential} etc.

We focus on exploring distributed approaches tailored for the aforementioned setting. In this context, the buyers are not required to disclose their private valuations to the market. Instead, the buyers update their expenditures/bids solely based on the randomly changing value of the items in each period. The market then allocates the items based on the expenditures/bids provided by the buyers. In light of this, we develop the bidding approach from the perspective of the buyers and aim to investigate the feasibility and effectiveness of achieving an efficient resource allocation and market equilibrium in an online environment without the need of revealing the buyers' valuations. 


\subsection{Main Contributions}
We summarize our contributions as follows:\\
\textbf{Decentralized approach for online Fisher Markets.} We present an online scheme, termed \emph{online proportional response}, in which each buyer dynamically adjusts their bids based on the allocation they received in the previous period and the stochastic value of the goods observed in the current period. The online proportional response serves as a practical and intuitive bidding mechanism, and ensures several desirable properties: 
\begin{enumerate}
    \item  buyers are not required to disclose their private information in competition. The market determines allocations based on buyers' bids rather than their valuations, and therefore, buyers have no chance to misreport their valuations to obtain better allocations;
    \item the market supplies consistently match the demands in each period, ensuring the continuous clearance of the market; 
    \item buyers' budget remains constant 
    in each period, resulting in the depletion of their total budgets by the conclusion end of the planning horizon. 
\end{enumerate}

\noindent\textbf{Connections to the online mirror descent.} We establish equivalence between the online proportional response method and the online mirror descent algorithm. Initially, we break down the Shmyrev (SH) convex program (\cite{shmyrev2009algorithm}), a variant of the EG convex program that also generates market equilibrium, from the perspective of online mirror descent (Lemma~\ref{lemma3.1}). Since the objective function of the SH convex program is not Lipschitz continuous, the standard online mirror descent approach is not applicable in this context. We address this challenge by selecting the Bregman divergence function accordingly and implementing a constant step-size (Proposition~\ref{pp3.1}). This results in a distributed structure for the algorithm, and the updates have a closed-form expression that aligns with the online proportional response (Lemma~\ref{lemma3.2}).

\noindent\textbf{Robust
performance guarantees in different stochastic worlds.}
We evaluate the performance of the online proportional response in terms of the \emph{fairness regret}, 
which is the cumulative difference in the objective value of a stochastic Shmyrev convex program between an online algorithm and an offline optimum. The fairness regret is developed from the definition of the weighted proportional fairness that have been considered in \cite{kelly1998rate, bertsimas2011price, bateni2022fair}, and aligns with the notion of regret commonly utilized in online decision-making contexts, providing a clear assessment of performance. In addition, we qualify the \emph{individual regret}, which is defined as the difference in total utility experienced by a buyer between the online proportional response and the offline market equilibrium, where the latter achieves the optimal Nash social welfare in the offline setting. 

We analyze performances in a view of the online mirror descent algorithm and assume that the buyers have access to the noisy gradient (true gradient with a random noise), which is a typical setting in the literature of online learning (\cite{beck2003mirror, nedic2014stochastic}, etc.). We solve a more general class of online convex optimization problems \emph{without relying on the Lipschitz continuous condition but with merely ``relative smoothness"}, which is possibly of {independent interest}. 
Suppose that the online proportional response is employed over a total of $T$ time periods. In each period, the quantities of each available item and the total expenditures of all buyers are normalized to one. We present the performances under various input (noise) models without the need to know in advance which input model we are facing and exemplify the robustness of proportional response under many stochastic worlds.

When the input (noise) is independently and identically drawn from a distribution that is unknown to the buyers and market, the online proportional response achieves an upper bound of $\log MN$ on the fairness regret, which depends specifically on the problem size and logarithmically scales with the number of items and buyers but \emph{not} increasing with the time horizon $T$. Regarding the individual regret, it has an upper bound of $\mathcal{O}(\sqrt{T})$, and the bounded growth rate indicates that individual regret grows sublinearly with respect to $T$. 

In order to reflect real-world scenarios more accurately, we also consider various \emph{nonstationary} input models. 
\begin{itemize}
\item Firstly, we analyze the performance in a scenario where an adversary introduces corruption in each period, but the input remains independent across periods. Remarkably, we establish that the online proportional response can recover the results achieved under i.i.d. input \emph{with small corruptions} (Theorem~\ref{nonstationary_independent}). 
\item Additionally, we investigate the performance in situations where the input exhibits dependence over time: we consider the \emph{ergodic} input model, where the dependency between data diminishes gradually as time progresses (Theorem~\ref{nonstationary_Ergodic}). 
\item Moreover, we examine the \emph{periodic} input model, where the data displays seasonality or cyclic patterns (Theorem \ref{nonstationary_Periodic}). 
\end{itemize}

\subsection{Related Work}
\paragraph{Distributed algorithms for Fisher markets.}
 One of the most popular distributed approaches in computing market equilibrium is the tatonnement process, which is a distributed price update process introduced by \cite{walras1900elements}. This process involves updating prices based on the dynamics of supply and demand. When the demand from buyers exceeds the supply of an item, its price increases; if the demand is insufficient, the price decreases. Another popular distributed algorithm proposed in \cite{zhang2011proportional} is the proportional response (PR) dynamics. In the PR dynamics, buyers update their expenditures in proportion to the contribution of each item to their current utility. \cite{birnbaum2011distributed} established a connection between the PR dynamics and the mirror descent algorithms using Bregman divergence. Based on this significant advancement in techniques, the authors of \cite{birnbaum2011distributed} further showed a linear convergence rate for linear Fisher markets. Convergence results for buyers with Constant Elasticity of Substitution (CES) utilities were derived in a subsequent work \cite{cheung2018dynamics}. 

 However, both the tatonnement process and PR dynamics, along with other existing distributed algorithms, require repeated interaction among buyers and the  supplier to achieve market equilibrium in a static market. In the case of the daily arrival of items without interruption, buyers do not engage in repeated interactions in the market. Instead, buyers tend to make daily purchase decisions based on ongoing factors and considerations. This unique characteristic of daily decision-making highlights the need for innovative approaches that accommodate the dynamics of online markets and the evolving preferences of buyers.

\paragraph{Online Fisher markets.} Online Fisher markets represent a type of online resource allocation problem that operates in a competitive environment. In a recent study (\cite{gao2021online}), the authors examined an online Fisher market where a single unit item arrives in each period. Their work focused on the dual formulation of the Eisenberg-Gale (EG) convex program. The allocation of the items for each period was determined using a first-price auction, where the buyer with the highest bid would receive the entire item. The authors adopted the convergence theorems of dual averaging and analyzed the individual buyer's regret and envy in a stationary setting. The nonstationary setting was subsequently explored in \cite{liao2022nonstationary}, which focused on analyzing the performance of dual averaging in such dynamic environments. Moreover, for the adversarial setting of online Fisher markets, where items arrive sequentially, the competitive analysis (\cite{borodin1998online}) of the proposed algorithms was investigated in studies \cite{azar2016allocate} and \cite{banerjee2022online}. These works provided insights into the competitive performance of algorithms under adversarial conditions. A different online setting, where buyers arrive sequentially in the market, was considered in studies \cite{sinclair2022sequential} and  \cite{manshadi2021fair}. These studies focused on designing policies that minimize fairness objectives. However, these works assumed that the preferences of the arriving buyers were known to the market. The issue of privacy preservation in such scenarios was addressed in \cite{jalota2022online}. \cite{kolumbus2023asynchronous} considered a scenario where an adversary selects a subset of the buyers who will update bids in ever period such that these buyers act asynchronously.

In our study, we explore a different setting within the realm of online Fisher markets. The sequential arrival of items has been studied in works such as \cite{azar2016allocate}, \cite{banerjee2022online}, \cite{gao2021online}, and \cite{liao2022nonstationary}. Similarly, the sequential arrival of buyers has been explored in \cite{sinclair2022sequential} and \cite{manshadi2021fair}. However, our focus is on scenarios where the same items are frequently restocked. Specifically, the same items will be sold out at the end of each period and restocked at the beginning of the next period. 
Due to the items being frequently restocked, it necessitates novel approaches to achieving market equilibrium and optimizing allocation efficiency.

\paragraph{Online resource allocation problems.} 
In the online resource allocation problem, customer requests for resources arrive sequentially, and decision-makers need to make allocation decisions upon the arrival of each request, consuming resources and obtaining rewards. Applications being studied include internet advertising (\cite{mehta2007adwords}) and online bidding with budget constraints (\cite{balseiro2019learning}). Recently, a study (\cite{balseiro2022best}) investigated online resource allocation problems within a data-driven setting. The authors proposed the dual mirror descent algorithm, which can serve as a pacing strategy to adjust multipliers periodically. A similar problem was examined in \cite{li2022online}, where the authors assumed the linearity of the objective function while acknowledging the potential non-concavity of the former. 

While minimizing regret, there is an issue of constraints violation associated with the aforementioned algorithms for the online allocation problems. Constraints violation occurs when the algorithm's decision leads to the unavailability of resources for allocation before the end of the planning horizon. In our work, the proposed algorithm can mitigate this issue by adopting the proportional response approach, which shares similarities with the \emph{multiplicative weights update} algorithm used in online learning. This ensures that the total bids in each period remain constant, preventing the buyers from depleting their budgets, i.e., the resource of buyers, before the end of the planning horizon.

\paragraph{Fair allocation.}
The issue of fairness in resource allocation problems has been extensively explored in welfare economics. In \cite{bertsimas2011price}, the authors delved into two prevalent notions of fairness: max-min fairness and proportional fairness. These criteria are widely used, and the study examined the loss in utilitarian efficiency resulting from ensuring a certain level of fairness. The authors also showed the connections between proportional fairness and the EG convex program. The seminal work \cite{varian1974equity} pointed out finding a fair allocation is the \emph{competitive equilibrium from equal incomes} solution which maximizes the Nash social welfare objective. The Nash social welfare objective provides a balance between fairness and efficiency and can be maximized by the solution of EG convex program. Thus, we analyze the Nash social welfare objective to demonstrate the proportional fairness, a perspective shared by various works, including  \cite{banerjee2022online},\cite{manshadi2021fair},\cite{sinclair2022sequential} etc.

\paragraph{Online learning algorithms.}
From a technical view, our designed algorithm lies in the context of the online learning algorithm. Most of the existing work builds on a general assumption that the gradient of the objective function is Lipschitz continuous, see \cite{nedic2014stochastic}. As far as we are aware, without this requirement, a deterministic setting was first introduced in \cite{birnbaum2011distributed} and offered a linear convergence rate by using the mirror descent with Bregman divergence. In \cite{lu2018relatively}, the authors generalized the required conditions for obtaining linear convergence rates under non-Lipschitz continuity. In this work, we provide regret-like results for online mirror descent with Bregman divergence in the adversarial setting. The premise of Lipschitz continuous gradient plays a crucial role in analyzing the regret in online learning algorithms. When the objective function is convex and satisfies the Lipschitz condition, the regret of an algorithm exhibits an order of $\mathcal{O}(\sqrt{T})$, where $T$ represents the number of periods or rounds. This indicates that the algorithm's performance improves with more iterations. On the other hand, when the objective function is strongly convex, the regret has a tighter bound and decreases logarithmically with~$T$, yielding an order of~$\mathcal{O}(\log T)$. This suggests that strongly convex objective functions enable faster convergence and reduced regret over time. It is worth noting that the theoretical lower bounds of regret have also been established in the literature, providing fundamental limits on the performance that can be achieved. Refer to the comprehensive introduction by Hazan et al. (~\cite{hazan2016introduction}) for further insights and in-depth understanding of these concepts. 

In addition to regret, in multi-agent learning, the last iterate convergence also attracts considerable attention. \cite{mertikopoulos2019learning} examined that in a strongly monotone game, the players' actions would converge to the Nash equilibrium of the underlying game if all players update their actions according to the dual averaging algorithm in a distributed manner. Lately, \cite{zhou2021robust} proved the last-iterate convergence rate in a strongly monotone game. Due to the special structure of the Shmyrev convex program, which is not a strongly monotone game, \textcolor{blue}{it is hard to give the results of the last-iterate convergence rate for the bids.}


\subsection{Notations}
In the sequel, we use $\mathcal{O}(\cdot)$ notation to suppress constants and the natural logarithm (base $e$), the lowercase letters to denote scalars, such as $x\in \mathbb{R}$, the boldface letters to denote vectors, such as $\mathbf{x} \in \mathbb{R}^n$, and $\mathbb{R}_+$ to denote the set of nonnegative reals. We denote by $\langle \mathbf{x},\mathbf{y} \rangle =\mathbf{x}^{\top}\mathbf{y}=\sum_{i =1}^n x_iy_i$ 
the inner product between vectors $\mathbf{x}$ and $\mathbf{y}$, where $x_i$ is the $i$-th coordinate of $\mathbf{x}$. Given a norm $\Vert \cdot \Vert$ of vectors,  the dual norm is defined as $\Vert \mathbf{y}\Vert_\ast=\max_{\Vert \mathbf{x} \Vert \leq 1 } \langle  \mathbf{x},\mathbf{y} \rangle$. For the standard Euclidean norm ($\ell_2$-norm), $\Vert \mathbf{x} \Vert=\Vert \mathbf{x} \Vert_2=\sqrt{\langle \mathbf{x},\mathbf{x}\rangle}$ and $\Vert \mathbf{x} \Vert_\ast =\Vert \mathbf{x} \Vert_2$. {For the {$\ell_1$-norm}, $\Vert \mathbf{x} \Vert_1=\sum_{i=1}^n \vert x_i \vert$ and $\Vert \mathbf{x} \Vert_{1\ast} =\Vert \mathbf{x} \Vert_\infty =\max_{1\leq i\leq n} \vert x_i \vert$.} 

\section{Preliminaries}
We first briefly review the basic components in the static linear Fisher markets and the related convex programs for characterizing the market equilibrium. Then, we introduce the online setting for Fisher markets. 
\subsection{Static (Offline) Linear Fisher Markets}
In a linear Fisher markets model, there are $N$ buyers and $M$ perfectly divisible items. We use $i \in \mathcal{N}:=\left\lbrace 1,2,\ldots,N \right\rbrace$ to index the buyers and $j \in \mathcal{J}: =\left\lbrace 1,2,\ldots,M \right\rbrace$ to index the items. Each  item $j \in \mathcal{M}$ has a supply that is normalized to one unit, and the value of each buyer $i \in \mathcal{N}$ towards one unit of each item $j$ is $v_{ij}\in \mathbb{R}_+$. 
Each {buyer~$i \in \mathcal{N}$} has a budget $B_i \in \mathbb{R}_+$.  Given any price vector {$\mathbf{p}=(p_1,\ldots,p_M) \in \mathbb{R}_+^M$}, the demand (or allocation) of each buyer~$i$ is a bundle of items $\mathbf{x}_i=(x_{i1},\ldots,x_{iM}) \in \mathbb{R}_+^M $ that maximizes her linear additive utility function $u_i:\mathbb{R}_+^M \mapsto  \mathbb{R}_+$ under the budget permission:
\begin{align}\label{single}
    \mathbf{x}_i:= \argmax\limits_{\mathbf{x}^\prime_i} \left\lbrace u_i(\mathbf{x}^\prime_i)=\sum_{j \in \mathcal{M}} v_{ij}x^\prime_{ij}: {\mathbf{p}^{\top} \mathbf{x}^\prime_i} \leq B_i, x^\prime_{ij} \geq 0, j \in \mathcal{M}\right\rbrace.
\end{align} 
{A collection of demands $\mathbf{x}=(\mathbf{x}_1,\ldots,\mathbf{x}_N) \in \mathbb{R}^{M \times N}$.} 
\begin{definition}[Market equilibrium]
A market equilibrium  $(\mathbf{x}^\ast,\mathbf{p}^\ast)$ satisfies the following conditions:
\begin{enumerate}
    \item \textit{Buyer optimality}: The allocation of each buyer $i \in \mathcal{N}$ maximizes her utility subject to the budget constraint, i.e., for each buyer $i  \in \mathcal{N}$, $\mathbf{x}_i^\ast$ is the solution of (\ref{single}).
    \item \textit{Supply feasibility}: The market demand of item $j \in \mathcal{M}$ {does not} exceed its supply, i.e., for each item $j \in \mathcal{M}$, $\sum_{i \in \mathcal{N}} x_{ij}^\ast\leq 1$. 
    \item \textit{Market clearance}: If the market demand of item $j$ is equal to its supply, i.e., $\sum_{i \in \mathcal{N}} x_{ij}^\ast= 1$, then the item $j$ has a price $p_j^\ast>0$; otherwise, $p_j^\ast=0$ if $\sum_{i \in \mathcal{N}} x_{ij}^\ast\leq 1$ 
\end{enumerate}
\end{definition}
Under equilibrium prices, all items are sold, buyers {deplete} their budgets, and they have no incentive to change their expenditures on the items. Each buyer only spends money on the items with maximum \emph{bang-per-buck} ratio at price $\mathbf{p}$, i.e., purchases the item {$j:= \argmax_{j'} \left\lbrace \frac{v_{ij'}}{p_{j'}}: j'\in \mathcal{M}\right\rbrace$}. Note that $p_j^\ast=0$ implies the bang-per-buck ratio tends to infinity if $v_{ij}>0$, which contradicts $\sum_{i \in \mathcal{N}} x_{ij}^\ast \leq 1$ because all buyers prefer this item $j$ to the other items. Hence, if each item has a potential buyer ($v_{ij}>0$ for any $j$), then the equilibrium prices are positive and all items {must be} fully sold out.
The equilibrium allocations can be characterized via the Eisenberg-Gale (EG) convex program (\cite{eisenberg1959consensus}) as below:
{\begin{align}\label{EG}
    \max\limits_{\mathbf{x}\in\mathbb{R}^{M\times N}} &\sum_{i \in \mathcal{N}} B_i\log u_i, \\\notag
     \mbox{subject to } & u_i=\sum_{j \in \mathcal{M}} v_{ij}x_{ij},\,
      i \in \mathcal{N},\\\notag
     & \sum_{i \in \mathcal{N}} x_{ij} \leq 1,\,  j \in \mathcal{M},\\\notag
       & x_{ij}\geq 0, \,  i \in \mathcal{N}, j \in \mathcal{M}.\notag
\end{align}
}
The set of optimal solutions of (\ref{EG}) is equal to the set of equilibrium allocations $\mathbf{x}^\ast$ and the utilities $u_i^\ast$; the equilibrium price $\mathbf{p}^\ast$ are the dual variables (or multipliers) of the supply constraints $\sum_{i \in \mathcal{N}} x_{ij}^\ast \leq 1, {\forall j \in \mathcal{M}}$. Moreover, the equilibrium utilities $u_i^\ast$ and prices $\mathbf{p}^\ast$ are unique (\cite{eisenberg1959consensus}).  The objective of (\ref{EG}) is the budget-weighted geometric mean of buyers' utilities and related to the Nash social welfare (\cite{cole2017convex},\cite{branzei2017nash}). 
The equilibrium allocations $\mathbf{x}^\ast$ satisfy the following desirable properties:
\begin{enumerate}
    \item \textit{Pareto optimality}: For any allocation $\mathbf{x}^\prime \neq \mathbf{x}$ which makes the utility of some buyer $i$ increase, i.e., \[\sum_{j \in \mathcal{M}} v_{ij}x_{ij}^\prime > \sum_{j \in \mathcal{M}} v_{ij}x_{ij},\] it must have the utility of some other buyer $k$ decrease, i.e., \[\sum_{j \in \mathcal{M}} v_{kj}x_{kj}^\prime < {\sum_{j \in \mathcal{M}}} v_{kj}x_{kj}.\] 
    \item \textit{Envy-freeness}: For any pair of buyers $(i,k)$, buyer $i$ prefers her allocations to buyer $k$'s allocations, i.e.,  \[\frac{\sum_{j \in \mathcal{M}} v_{ij}x_{ij}}{B_i}\geq \frac{\sum_{j \in \mathcal{M}} v_{ij}x_{kj}}{B_k}.\]
    \item \textit{Proportionality}: For any buyer $i$, she prefers her allocations to the uniform allocation, i.e., \[\sum_{j \in \mathcal{M}} v_{ij}x_{ij} \geq \sum_{j \in \mathcal{M}} \frac{v_{ij}}{N}. \] 
\end{enumerate}

 \subsection{Online Linear Fisher Markets}
We consider an online linear Fisher market that operates over $T$ distinct periods.  The market involves a fixed population of $N$ heterogeneous buyers, and each buyer $i$ is endowed with {an initial total budget $B_{i,1} \in \mathbb{R}_+ $} at $t=1$. In each period $t$, the market is supplied with $M$ items, where the quantity of each item is normalized as one unit, and this supply is replenished at the beginning of each period. The market provides buyers with information about the items, which heterogeneously influences their valuations based on their specific {targeting criteria (such as quality, market trends)}. We denote the valuation of item $j$ privately observed by buyer $i$ in each period $t$ as  $v_{ij,t} \in \mathbb{R}_+$. In each period $t$, the market determines $x_{ij,t}$ the allocation of item $j$ to buyer $i$, and the item price $p_{j,t}$. We denote by {$u_{i}=\sum_{t=1}^T\sum_{j \in \mathcal{M}} v_{ij,t}x_{ij,t}$} the cumulative utility of buyer~$i$ over total~$T$ periods. In particular, we assume that the market has no information about the valuation of buyers and that each buyer does not know the number of competitors in the market and their valuations.

In an ideal scenario where the market has access to complete information on buyers' private item values and the target expenditures in each period, it can compute the market equilibrium and achieves many desirable properties. In practice,  buyers are often reluctant to reveal their private information. Additionally, buyers may exhibit strategic behavior by misreporting their  values to obtain more favorable allocations, leading to unfair allocations (\cite{branzei2014fisher}). While achieving market equilibrium is possible through certain distributed algorithms, these algorithms typically require constant interactions among the buyers and the market, and buyers who are often reluctant to engage in frequent interactions with the market and prefer to minimize their involvement. 
Buyers would possibly like to participate in a bidding {scheme}, where they only have to submit bid prices for items and the market decide the corresponding {item prices} and allocations. 

We now introduce the {online bidding scheme}. In each period $t$, each buyer $i$ submits an \emph{irrecoverable} bid price $b_{ij,t}$ for each item $j$ to the markets, and the sum of bids is equal to the averaged budget { $B_i$ per period, i.e., $B_i = \sum_{j \in \mathcal{M}}b_{ij,t}=B_{i,1}/T$. }
The market uses the trading post mechanism, which is first presented in \cite{shapley1977trade}, to set prices and allocations. Specifically, the market {sets} prices as $p_{j,t}= \sum_{i \in \mathcal{N}} b_{ij,t}$ for all item $j$ after receiving the bids from buyers, and then determines the allocation as $x_{ij,t}=b_{ij,t}/p_{j,t}$ for all $i,j$. Buyers do not need to observe other buyers' bids but can actually derive the prices or compute prices after receiving allocations, i.e., $p_{j,t}=b_{ij,t}/x_{ij,t}$. 
{We assume that the length of $T$ is known for all buyers and the summation of all buyers' budgets is normalized as $T$\footnote{The normalization does not affect the buyers' received utilities.}. }

Define the  history available at period $t$ to each buyer $i$ as 
$$\mathcal{H}_{i,t}:= \left\lbrace \left(\left(v_{ij,\tau},x_{ij,\tau},b_{ij,\tau},p_{j,\tau},\right)_{j \in \mathcal{M}}\right)_{\tau=1}^{t-1}, \left(v_{ij,t}\right)_{j\in \mathcal{M}} \right\rbrace,$$
for any $t\geq 2$, and $\mathcal{H}_{i,1}=\left\lbrace  (v_{ij,1})_{j\in \mathcal{M}} \right\rbrace$ for $t=1$. A bidding strategy for buyer $i$ generates a sequence of functions $\pi_i =\left\lbrace\pi_{i1},\ldots,\pi_{iT}) \right\rbrace$, where {$\pi_{i1}: \mathbb{R}^M_+ \mapsto \mathbb{R}_+^M$ and $\pi_{i,t}: \mathbb{R}^{4M(t-1)+{M}}_+ \mapsto \mathbb{R}_+^M$ for $t\geq 2$}, that maps available histories to bid price, that is
\begin{align*}
    &\mathbf{b}_{i,1}=\pi_{i1}(\mathcal{H}_{i,1}), t=1\\
    &\mathbf{b}_{i,t}=\pi_{it}(\mathcal{H}_{i,t}),  t=2,\ldots,T.
\end{align*}
A strategy is \emph{admissible} if it only depends on the information available to  buyers, or to say, the bidding strategy $\pi_i$ is measurable  with respect to the filtration $\mathcal{H}_{i,t}$.

\section{Proportional Response for Online Fisher Markets}
In this section, we first give a review of the mirror descent algorithm, which has demonstrated effectiveness in optimizing convex objectives in various optimization problems. Subsequently, we introduce the online version of the proportional response and establish its relationship with the mirror descent algorithm. Towards the end of this section, we define the performance metrics that will be used to evaluate the algorithm and outline the key assumptions underlying our analysis.

\subsection{Warm Up: Mirror Descent}
Consider the problem of minimizing a convex differentiable function $\varphi: \mathbb{R}^n \mapsto \mathbb{R}$ over the convex and compact constraint set $K \subset \mathbb{R}^n$, 
\begin{equation*}
     \min\;\varphi(\mathbf{p}), \mbox{ subject to } \mathbf{p} \in K.
\end{equation*}
We consider the mirror descent algorithm for solving the above problem. 
Define the first order approximation of $\varphi(\mathbf{p})$ at $\mathbf{q} \in K$ as
$$l_\varphi(\mathbf{p};\mathbf{q}):=\varphi(\mathbf{q})+\langle\nabla \varphi(\mathbf{q}),\mathbf{p}-\mathbf{q}\rangle,$$
where $\nabla \varphi(\mathbf{q})$ denotes the {gradient} of $\varphi(\cdot)$ at $\mathbf{q}$. A gradient $\nabla \varphi(\mathbf{q})$ is usually assumed to exist for $\mathbf{q} \in K$. Let $h:\mathbb{R}^n \mapsto \mathbb{R}$ be a continuously differentiable and strongly convex function over the set $K$ with a scalar $\mu_h>0$, i.e., $h(\mathbf{p}) \geq h(\mathbf{q})+\langle \nabla h(\mathbf{q}), \mathbf{p}-\mathbf{q}\rangle + \frac{\mu_h}{2} \Vert \mathbf{p}- \mathbf{q} \Vert^2 $. 
Given any $ \mathbf{p},\mathbf{q} \in K$, we define the \emph{Bregman divergence (distance)} function as
$$D_h(\mathbf{p},\mathbf{q}):=h(\mathbf{p})-h(\mathbf{q})-\langle \nabla h(\mathbf{q}),\mathbf{p}-\mathbf{q} \rangle.$$ 
Note that the Bregman divergence is not symmetric, i.e., $D_h(\mathbf{p},\mathbf{q})\neq D_h(\mathbf{q},\mathbf{p})$, but it has that $D_h(\mathbf{p},\mathbf{q}) \geq \frac{\mu_h}{2} \Vert \mathbf{p}- \mathbf{q} \Vert^2$ and  $D_h(\mathbf{p},\mathbf{q})=0$ if and only if $\mathbf{p}=\mathbf{q}$. 
The Bregman divergence satisfies the following relation:
{$$D_h(\mathbf{z},\mathbf{p})+D_h(\mathbf{p},\mathbf{q})-D_h(\mathbf{z},\mathbf{q})=\langle\nabla h(\mathbf{p})-\nabla h(\mathbf{q}),{\mathbf{p}-\mathbf{z}}\rangle, \mathbf{p},\mathbf{q},\mathbf{z} \in K,$$
}which {is} also called {the \emph{three-point identity}} and is widely used in literature  (\cite{beck2003mirror} \cite{nedic2014stochastic}).

\textbf{Mirror descents work iteratively.} It starts with an initial point~$\mathbf{p}_0 \in K$, sets step-size $\eta_t>0$, and apply the update rule as follows. 
\begin{equation}
    \mathbf{p}_{t+1}:=\argmin_{\mathbf{p} \in K} \left\lbrace  \eta_t l_\varphi(\mathbf{p};\mathbf{p}_t)+D_h(\mathbf{p},\mathbf{p}_t)\right\rbrace=\argmin_{\mathbf{p} \in K} \left\lbrace  \eta_t \langle\nabla \varphi(\mathbf{p}_t),\mathbf{p}-\mathbf{p}_t\rangle+D_h(\mathbf{p},\mathbf{p}_t)\right\rbrace, \forall t\geq 0. \label{standardMD}
\end{equation}

In particular, when $h(\mathbf{p})=\frac{1}{2} \Vert \mathbf{p}\Vert_2^2$ is strongly convex over $\Omega=\mathbb{R}^n$ with respect to the $\ell_2$-norm, the Bregman divergence is 
$D_h(\mathbf{p},\mathbf{q})=\frac{1}{2} \Vert \mathbf{p}-\mathbf{q}\Vert_2^2$. In this case, the update rule~(\ref{standardMD}) is equal to the {projected} gradient descent method, that is,
$$\mathbf{p}_{t+1}:=\Pi_{K} \left(\mathbf{p}_t-\eta_t\nabla \varphi(\mathbf{p}_t)\right),$$
where {$\Pi_K$} denotes the projection of $\mathbf{p}$ to the set $K$; if $h(\mathbf{p})=\sum_{i=1}^n(p_i\log p_i -p_i)$ is strongly convex over the probabilistic simplex $\Omega=\left\lbrace \mathbf{p} \in \mathbb{R}^n: \Vert \mathbf{p} \Vert_1=1 \right\rbrace$ with respect to the $\ell_1$-norm,  the Bregman divergence is $D_h(\mathbf{p}, \mathbf{q})=\sum_{i=1}^np_{i}\log(\frac{p_i}{q_i})$, that is the \emph{Kullback-Leibler (KL)} divergence between the discrete distributions $\mathbf{p}$ and $\mathbf{q}$. In this case,  the update rule (\ref{standardMD}) recovers the multiplicative weights update (\cite{arora2012multiplicative}):
$$\mathbf{p}_{t+1}:=\mathbf{p}_t {\odot} \exp(-\eta_t\nabla \varphi(\mathbf{p}_t)),$$
where $x \odot y = (x_jy_j)_{j=1}^n $ is the Hadamard  product of vectors~$x,y \in \mathbb{R}^n$. 

The online mirror descent  algorithm serves as an efficient algorithm to solve online convex optimization problems. The update rule is: in  period $t$, a decision maker receives a first-order feedback $\nabla \varphi_t(\mathbf{p}_t)$ after choosing $\mathbf{p}_t$, where $\nabla \varphi_t(\cdot)$ is chosen by an adversary and could be used in a mirror descent update (\ref{standardMD}) to choose $\mathbf{p}_{t+1}$ for next period $t+1$.

\subsection{Online Proportional Response}
Now we investigate the connection between the online mirror descent and the proportional response in our ``online" Fisher markets. We first introduce an alternative convex program proposed by Shmyrev (\cite{shmyrev2009algorithm}) that also captures market equilibrium of static linear Fisher markets, that is:
{
 \begin{align}\label{shmyrev}
    \min &\sum_{i \in \mathcal{N}}\sum_{j \in \mathcal{M}} -b_{ij} \log v_{ij}+\sum_{j \in \mathcal{M}}p_j\log p_j, \\\notag
  \mbox{subject to } & \sum_{i \in \mathcal{N}}b_{ij}=p_j,  j \in \mathcal{M},\\\notag
 & \sum_{j \in \mathcal{M}}b_{ij}=B_i,  i \in \mathcal{N},\\\notag
  & b_{ij}\geq 0,  i \in \mathcal{N}, j \in \mathcal{M}, \notag
 \end{align}}where the corresponding allocation $x_{ij}$ is given by $b_{ij}/p_j$. One advantage of this program is that it solely revolves around the bid price vector $\mathbf{b}$, which is determined by the buyers themselves. This buyer-oriented nature of the program serves as a motivation for our study of the bidding mechanism in online Fisher markets. By focusing on the bidding decisions made by individual buyers, we aim to design a bidding algorithm that enables buyers to express their preferences and values for the available items effectively.

We denote by $\mathbf{b}_t $ the vector of bid price $b_{ij,t}$,  and $\mathbf{p}_t$ the vector of $p_{j,t}$. 
For each buyer $i$ at period $t$, the total bid is equal to the averaged budget per period {(i.e., $B_i=B_{i,1}/T$)}.  
We assume that each buyer updates  bid price $\mathbf{b}_{i,t}$ based on the available history information $\mathcal{H}_{i,t}$ as we defined previously in Section 2.2. Most importantly, \emph{buyers do not need to share their information or bid prices with each other}. Consequently,  each buyer can only observe the ex-post price of the items, which is represented by the price vector $\mathbf{p}_{t-1}$,  available at period $t$. These practical considerations of limited information sharing and reliance on ex-post price information align with real-world scenarios where buyers make bid price adjustments based on their individual perspectives and available data.

To describe the market randomness in each period $t$, we  define the objective of Shmyrev convex program as
$$\varphi(\mathbf{b},\mathbf{\varepsilon}_t):=-\sum_{i \in \mathcal{N}}\sum_{j \in \mathcal{M}} b_{ij} \log v_{ij,t}+\sum_{j \in \mathcal{M}} p_{j}\log p_{j},$$
the price $p_j=\sum_{i \in \mathcal{N}} b_{ij}$ for $j \in \mathcal{M}$, the constraint set as $K_i:= \left\lbrace \mathbf{b}_i: \sum_{j \in \mathcal{M}} b_{ij}=B_i, b_{ij}\geq 0, \forall i\in\mathcal{N} \right\rbrace$, the union set as $K:= \left\lbrace {\bigcup_{i=1}^N K_i} \right\rbrace$, and the gradient  $\nabla \varphi(\mathbf{b},\varepsilon_t)=(1-\log (v_{ij,t}/p_{j}))_{ij}$ for $i \in \mathcal{N}, j\in \mathcal{M}$. Here, $\varepsilon_t \in \mathbb{R}^{M \times N}$ is the noise of the gradient, representing the market randomness that affects $v_{t} = \tilde{v}_{t}(\epsilon_t)$, where $\tilde{v}_t(\epsilon_t): \mathbb{R}^{M\times N} \mapsto \mathbb{R}^{M\times N}$ is considered a valuation function depending on stochastic input. Given any $\mathbf{b}_t \in K$, the first order approximation of $\varphi(\mathbf{b}_t,\varepsilon_t)$ at $\mathbf{b}_{t-1} \in K$ is 
$$\ell_\varphi(\mathbf{b}_t,\mathbf{b}_{t-1},\varepsilon_t):=\varphi(\mathbf{b}_{t-1},\varepsilon_t)+\langle\nabla \varphi(\mathbf{b}_{t-1},\varepsilon_t),\mathbf{b}_t-\mathbf{b}_{t-1}\rangle.$$

The next lemma  establishes the connection between the mirror descent update (\ref{standardMD}) and the objective function $\varphi(\mathbf{b}_t,\varepsilon_t)$.  
\begin{lemma}\label{lemma3.1}
Let $h(\mathbf{p})=\sum_{j \in \mathcal{M}}(p_j\log p_j -p_j)$, and its Bregman divergence is  $D_h(\mathbf{p}, \mathbf{q})=\sum_{j \in \mathcal{M}}p_{j}\log(\frac{p_j}{q_j})$, then for any $\mathbf{b}_t,\mathbf{b}_{t-1} \in K$, it has that
$\varphi(\mathbf{b}_t,\varepsilon_t)=\ell_\varphi(\mathbf{b}_t,\mathbf{b}_{t-1},\varepsilon_t)+D_h(\mathbf{p}_t,\mathbf{p}_{t-1}).$
\end{lemma}

Clearly, if the step-size $\eta_t$ in standard mirror descent update rule (\ref{standardMD}) is set as $1$ for any $t$, then minimizing $\ell_\varphi(\mathbf{b}_t,\mathbf{b}_{t-1},\varepsilon_t)+D_h(\mathbf{p}_t,\mathbf{p}_{t-1})$ is equivalent to minimize $\varphi(\mathbf{b}_t,\varepsilon_t)$.

An important observation is that the gradient $\nabla \varphi(\mathbf{b}_{t-1},\varepsilon_t) = (1-\log(v_{ij,t}/p_{j,t-1}))_{ij}$ incorporates a logarithmic function. The potential issue arises when $p_{j,t-1}$ equals zero, rendering the logarithm undefined, which results in a unbounded gradient and implies that $\varphi(\mathbf{b},\varepsilon)$ is \emph{not} Lipschitz continuous. Convergence guarantees for the mirror descent heavily rely on the assumption of Lipschitz continuity, as extensively discussed in \cite{beck2003mirror,nedic2014stochastic} and other related works. To address the issue of non-Lipschitz continuity, we first introduce a condition called \emph{relative smoothness}, which is defined as follows:  
\begin{definition}[Relative smoothness] For any $\varepsilon$, $\varphi(\cdot,\varepsilon)$ is relative smooth to $h(\cdot)$ on $K$ if for any $\mathbf{b},\mathbf{b}^\prime \in  K$, there is a positive constant $L$ for which
$$\varphi(\mathbf{b},\varepsilon) \leq \ell_\varphi(\mathbf{b},\mathbf{b}^\prime,\varepsilon)+L\cdot D_h(\mathbf{b},\mathbf{b}^\prime).$$
\end{definition}
This condition, having been explored in \cite{bauschke2017descent,lu2018relatively}, allows for the establishment of a sublinear $\mathcal{O}(1/T)$ rate of convergence in the deterministic setting ($\varepsilon=0$) if the update rule is designed based on this condition. To this end, using the fact that $D_h(\mathbf{p}_t,\mathbf{p}_{t-1}) \leq D_h(\mathbf{b}_t,\mathbf{b}_{t-1})$ (see Lemma 7 in \cite{birnbaum2011distributed}), we show the fulfillment of the relative smoothness condition as follows:
\begin{equation}\label{relatives_smoothness}
\varphi(\mathbf{b}_t,\varepsilon_t)=\ell_\varphi(\mathbf{b}_t,\mathbf{b}_{t-1},\varepsilon_t)+D_h(\mathbf{p}_t,\mathbf{p}_{t-1}) \leq \ell_\varphi(\mathbf{b}_t,\mathbf{b}_{t-1},\varepsilon_t)+D_h(\mathbf{b}_t,\mathbf{b}_{t-1}),  \mathbf{b}_t, \mathbf{b}_{t-1} \in K.
\end{equation}
It is worth noting that prior works concerning the relative smoothness condition have mainly been conducted in the deterministic setting, while our context pertains to the online (stochastic) setting. 
The following proposition establishes the update rule of online mirror descent algorithm without relying on Lipschitz continuity.

\begin{proposition}\label{pp3.1}
If $\eta_{i,t}=1$ for any $t \geq 2$  and buyer $i \in \mathcal{N}$, we have  
\begin{align}\label{updaterule}
    \mathbf{b}_t : &=\mathop{\argmin}_{\mathbf{b} \in K}  \left\lbrace \ell_\varphi(\mathbf{b},\mathbf{b}_{t-1},\varepsilon_t)+ D_h(\mathbf{b},\mathbf{b}_{t-1})\right\rbrace \\ \notag
    &=\mathop{\argmin}_{\mathbf{b} \in K}  \left\lbrace \langle\nabla \varphi(\mathbf{b}_{t-1},\varepsilon_t),\mathbf{b}-\mathbf{b}_{t-1}\rangle+ D_h(\mathbf{b},\mathbf{b}_{t-1})\right\rbrace \\ \notag
    &=\mathop{\argmin}_{\mathbf{b} \in K} \sum_{i \in \mathcal{N}}\sum_{j \in \mathcal{M}}(1-\log(v_{ij,t}/p_{j,t-1}))(b_{ij}-b_{ij,t-1})+\sum_{i \in \mathcal{N}}\sum_{j \in \mathcal{M}} b_{ij} \log \left(\frac{b_{ij}}{b_{ij,t-1}}\right).
\end{align}
\end{proposition}
The proof of the above proposition is trivial, the last equation comes from the fact that the gradient $\nabla \varphi(\mathbf{b}_{t-1},\varepsilon_t) $ is $(1-\log(v_{ij,t}/p_{j,t-1}))_{ij}$ for $i \in \mathcal{N}, j\in \mathcal{M}$.
 Proposition \ref{pp3.1} leads to a distributed bid price update mechanism, that is, each buyer can follow the update rule 
\begin{align} \label{singleupdaterule}
    \mathbf{b}_{i,t} :&=\mathop{\argmin}_{\mathbf{b}_i \in K_i} \sum_{j \in \mathcal{M}}(1-\log(v_{ij,t}/p_{j,t-1}))(b_{ij}-b_{ij,t-1})+\sum_{j \in \mathcal{M}} b_{ij} \log \left(\frac{b_{ij}}{b_{ij,t-1}}\right) \notag \\
    &=\mathop{\argmin}_{\mathbf{b}_i \in K_i}  \left\lbrace \langle\nabla \varphi(\mathbf{b}_{i,t-1},\varepsilon_{i,t}),\mathbf{b}_{i}-\mathbf{b}_{i,t-1}\rangle+ D_h(\mathbf{b}_i,\mathbf{b}_{i,t-1}) \right\rbrace,
\end{align}
where $\nabla \varphi(\mathbf{b}_{i,t-1},\varepsilon_{i,t})$ is the vector that consists of $(\nabla \varphi(\mathbf{b}_{t-1},\varepsilon_{t}))_{ij}$ for $j \in \mathcal{M}$. We observe that buyer $i$'s bids produced by the update rule (\ref{singleupdaterule}) depend on the prices and bids that made on the previous periods and coincides with our information structure, that is, $\mathbf{b}_{i,t}$ is based on the available history information $\mathcal{H}_{i,t}$. In other words,
the distributed structure ensures that each buyer can adopt the online mirror descent to update the bid price $\mathbf{b}_{i,t}$ {in} each period without disclosing the valuations. The formal descriptions of the algorithm for each buyer $i$ are given {in ALGORITHM~\ref{alg:one}}.


\begin{algorithm}
\SetKwInOut{Initial}{initial}
	\Initial{Each buyer $i$ sets initial bid price $b_{ij,1}=\frac{B_i}{MN}$, then receives $x_{ij,1}$ and computes the price $p_{j,1}=\frac{ b_{ij,1}}{x_{ij,1}}$ for all item $j\in \mathcal{M}$ .}
		\For{$t=2,\ldots, T$}{ 
            \For{each buyer $i \in \mathcal{N}$}{
             1. observes $v_{ij,t}$, computes the gradients $\nabla \varphi(\mathbf{b}_{t-1},\varepsilon_t)$, where each component is:
            $$(\nabla \varphi(\mathbf{b}_{t-1},\varepsilon_t))_{ij}=1-\log\frac{v_{ij,t}}{p_{j,t-1}}, j\in \mathcal{M};$$
            2. updates her bid price vector:
            \begin{align}
              \mathbf{b}_{i,t}=\mathop{\argmin}_{\mathbf{b}_i \in K_i} \left\lbrace  \langle \nabla \varphi(\mathbf{b}_{i,t-1},\varepsilon_{i,t}),\mathbf{b}_{i} {- \mathbf{b}_{i,t-1}} 
              \rangle+\sum_{j \in \mathcal{M}} b_{ij}\log \left(\frac{b_{ij}}{b_{ij,t-1}}\right)\right\rbrace. \label{OMDupdate}
            \end{align}
           3. receives allocation $x_{ij,t}$ and computes the price $p_{j,t}=\frac{b_{ij,t} }{x_{ij,t}}$ for $j \in \mathcal{M}$.}   
        }
    	\caption{Online mirror descent for online linear Fisher markets}
	\label{alg:one}
\end{algorithm}
This algorithm differs from traditional online learning processes, where buyers typically place bids for period $t+1$ in the end of period $t$. However, we believe it is more reasonable for buyers to bid after observing the values of the items.


Now, we derive a closed-form expression for updating the bid price.
\begin{lemma}\label{lemma3.2}
Suppose that each buyer $i$ use the algorithm to update bid price $\mathbf{b}_{i,t}$ for any $t=1,\ldots,T$, then 
\begin{equation} \label{onlinePR}
     b_{ij,t}=B_{i}\frac{v_{ij,t}x_{ij,t-1}}{\sum_{j \in \mathcal{M}} v_{ij,t}x_{ij,t-1}}. \quad \mbox{(Online proportional response)}  
\end{equation}
\end{lemma}
This above lemma implies that each buyer can update the bid price $\mathbf{b}_{i,t}$ by directly following the closed-form expression above instead of solving a constrained optimization problem per period. Therefore, each buyer can directly update the bid price according to~(\ref{onlinePR}), which is an online variant of the proportional response in a deterministic linear Fisher market~(\cite{zhang2011proportional}), and we {call the update rule in (\ref{onlinePR})} \emph{online proportional response}. Moreover, it can be easily seen that the budget consumption of each buyer $\sum_{j \in \mathcal{M}}b_{ij,t}=B_i$ is constant for any~$t$, so the budget of each buyer would be used up at period~$T$. {Hence, the \emph{constraints violation} issue~(see \cite{balseiro2019learning, balseiro2022best, li2022online}, etc), which describes the loss due to resources
being used up before the~$T$ period, no longer exists here.}

\subsection{Assumptions and Performance Metrics}
Now we introduce the performance metric of the online proportional response. At first, let $\mathbf{b}^\ast$ the optimal bids that minimizes the following program with $\Phi (\mathbf{b}):=\mathbb{E}_{\varepsilon\sim\mathcal{P}}\left[\varphi(\mathbf{b},\varepsilon) \right]$
\begin{align}\label{SHHindsight}
     \min \;&\Phi (\mathbf{b})  \\\notag
     \mbox{subject to} & \sum_{i \in \mathcal{N}} b_{ij}=p_j,  j \in \mathcal{M}\\\notag
     & \sum_{j \in \mathcal{M}} b_{ij}=B_i,  i \in \mathcal{N}, \notag
\end{align}
where the expectation is taken with respect to $\varepsilon$, which is drawn from a distribution $\mathcal{P}$. {We term (\ref{SHHindsight}) as the stochastic Shmyrev convex program.} 
Therefore, we characterize $\mathbf{b}^\ast$ the optimal bids  and the corresponding pair $(\mathbf{x}^\ast, \mathbf{p}^\ast)$ made by $\mathbf{b}^\ast$ constitutes a market equilibrium. Consequently, the market equilibrium allocation $\mathbf{x}^\ast$ guarantees a Pareto optimal, envy-free, and proportional allocation.
\noindent\paragraph{Proportional fairness.} Before introducing performance metrics, we begin with a brief introduction to the definition of the \emph{weighted proportional fairness}. The allocation $\mathbf{x}$ is \emph{weighted proportional fair} if the aggregate proportional change is less than or equal to $0$ compared to the utility of any other feasible allocation $\mathbf{x}^\prime$, that is 
$$\sum_{i \in \mathcal{N}} w_i\frac{u_i(\mathbf{x}_i^\prime)-u_i(\mathbf{x}_i)}{u_i(\mathbf{x}_i)}\leq 0. $$
Note that $w_i$ is a weight related to the budget of buyer $i$. As the budgets of each buyer are scaled as equal, $w_i$ can be considered as 1 for all $i$. This fairness notion has been widely adopted in a large body of literature and it is known that finding a proportionally fair allocation is   equivalent to finding a solution that maximizes the Nash social welfare objective (see \cite{kelly1998rate}, \cite{bertsimas2011price}, etc):
$$\max \prod_{i=1}^N u_i(\mathbf{x}_i)^{w_i}.$$
Naturally, a market equilibrium allocation is a proportional fair allocation.
Moreover, in the case of equal budgets (the weights are equal), maximizing the Nash social welfare objective yields a market equilibrium which is known in economics as \emph{competitive equilibrium from equal incomes} (\cite{varian1974equity}), and our problem also lies in this context. 

Since the optimal solution of the EG convex program also maximizes the Nash social welfare, the difference between the objective value of the EG convex program at two different allocations evaluates the fairness. 
Define $\Psi (\mathbf{b})=\mathbb{E}\left[-\sum_{i \in \mathcal{N}} B_i\log \left(\sum_{j \in \mathcal{M}}\frac{{v_{ij}b_{ij}}}{p_j}\right) \right]$, the opposite of the expectation of a random EG convex objective function, and then we have the following lemma.
\begin{lemma} \label{lemma3.3}
   Suppose $\mathbf{b}^\ast$ minimizes (\ref{SHHindsight}). Then for any {$\mathbf{b} \in K$}, it has that
   $$\Psi (\mathbf{b}) -\Psi (\mathbf{b}^\ast) \leq \Phi (\mathbf{b}) -\Phi (\mathbf{b}^\ast). $$
\end{lemma}
The above lemma states that the expected difference between the objective value of the EG convex program at allocations $\mathbf{x}$ (determined by the  bids $\mathbf{b}$) and allocations $\mathbf{x}^\ast$ (determined by the bids $\mathbf{b}^\ast$)  is upper bounded  by the objective value of the SH convex program at allocations $\mathbf{x}$ and allocations $\mathbf{x}^\ast$. Therefore, we consider the expected difference of the objective value of the SH convex program at two different allocations as the fairness metric. 

\subsubsection*{Performance metrics} 
The performance of an admissible strategy $\pi=(\pi_i,\ldots,\pi_n)$ is examined in two folds: \\
(i) \emph{fairness regret}: the cumulative difference in the expected objective value of (\ref{SHHindsight}) between an online algorithm and a hindsight equilibrium allocation:
    $$\mathcal{F}^\pi = \sum_{t=1}^T\mathbb{E}_{\varepsilon_1,...,\varepsilon_T}[\varphi(\mathbf{b}_t,\varepsilon_t)]- \mathbb{E}_\varepsilon[\varphi(\mathbf{b}^\ast,\varepsilon)] =\sum_{t=1}^T \Phi (\mathbf{b}_t) -\Phi (\mathbf{b}^\ast).$$
Here we use the allocations $\mathbf{x}^\ast$ (determined by the  bids $\mathbf{b}^\ast$) as a fair allocation benchmark. Recall that $\mathbf{b}^\ast$ is the optimal solution of the program (\ref{SHHindsight}), the performance metric (i) coincides with the classic definition of regret in the online convex optimization (\cite{hazan2016introduction}). 
\\
(ii) \emph{individual buyer's regret}: the difference in utilities attained by an admissible strategy $\pi_i$ and attained by the equilibrium (fair) allocation $\mathbf{x}^\ast$:
$$\mathcal{R}_i^\pi=\sum_{t=1}^T \mathbb{E}[v_{ij,t}x_{ij}^\ast -v_{ij,t}x_{ij,t}] =\sum_{t=1}^T \mathbb{E} \left[u_i^\ast -u_{i,t}\right].$$

We next introduce some constants.
\begin{assumption} \label{assumption2}
For every period $t=1,\ldots,T$ and every $i,j$,
\begin{enumerate}
    \item The value $v_{ij,t} >0$ and there exists $\underline{v}_i \in \mathbb{R}_+ $  and $\overline{v}_i \in \mathbb{R}_+ $ such that $0<\underline{v}_i \leq v_{ij,t} \leq \overline{v}_i$.  
    \item There exist $\underline{u}_i \in \mathbb{R}_+ $  and $\overline{u}_i \in \mathbb{R}_+ $ such that
    the utility functions satisfy $0<\underline{u}_i \leq u_{i,t} \leq \overline{u}_i$.
    \item There exist $\underline{p}_j~\in \mathbb{R}_+ $ and $\overline{p}_j \in \mathbb{R}_+ $ such that $0<\underline{p}_j \leq p_j^{\ast} \leq \overline{p}_j$. 
    \item  There exist $\underline{\varepsilon}_{ij} \in \mathbb{R}_-$ and $\overline{\varepsilon}_{ij} \in \mathbb{R}_+$ such that $\underline{\varepsilon}_{ij} \leq \varepsilon_{ij,t} \leq \overline{\varepsilon}_{ij}$.
\end{enumerate}
\end{assumption}
The requirement of the positive value ensures that in each period $t$, each item has at least one potential buyer; moreover, it guarantees that online proportional response has the ability to generate bid price for all items, otherwise, if $v_{ij,\tau}=0$ at period $\tau$, then $b_{ij,t}=0$ for all $t>\tau$. The upper bounds and lower bounds on the utility functions of buyers, the objective of the SH convex program, and the equilibrium price in hindsight are not required in implementing algorithm but these appear in the performance analysis. Specifically, we can derive that {$\overline{u}_i=(B_i\overline{v}_i)/{\underline{p}}$}, and define $\overline{v} =\max_i \overline{v}_i$ and $\underline{p} =\min_j \underline{p}_j$. In addition, the noise $\varepsilon_t$ belongs to the finite support set $\mathcal{S}:= \left\lbrace\underline{\varepsilon}_{ij},\overline{\varepsilon}_{ij}) \right\rbrace_{i \in \mathcal{N},j\in \mathcal{M}}$.
\begin{proposition} \label{generic_upper_bound}
Suppose that each buyer~$i$ 
{follows ALGORITHM~1} to update bid price $\mathbf{b}_{i,t}$ for any $t$, then for any buyer $i \in \mathcal{N}$, we have that
\begin{equation}\label{regretbound}
    \mathbb{E} \left[u_i^\ast -u_{i,t}\right]\leq \frac{\sqrt{2}\overline{v}}{\underline{p}} \sqrt{ \Phi (\mathbf{b}_t) -\Phi (\mathbf{b}^\ast) }.
\end{equation}

\end{proposition}
The inequality (\ref{regretbound}) shows that the individual buyer's regret at period $t$ is associated with  a multiplicative factor $\overline{v}/\underline{p}$ and the gap between the objective of the SH convex program at bid price $ \mathbf{b}^{\ast}$ and the one at the bid price $\mathbf{b}_{t}$ generated by the strategy $\pi$. 

\section{Stationary Input} \label{IIDinput}
In this section, we focus on the case of stationary input. In each period~$t$, buyers receive samples  $\varepsilon_t$ that are independently and identically distributed, drawn from a probability distribution $\mathcal{P}$ belonging to the space $ \Delta(\mathcal{S})$, where $\Delta(\mathcal{S})$ represents the set of all probability distributions over the support set $\mathcal{S}$. We assume that both the buyers and the market are unaware of the true distribution  $\mathcal{P}$. Now we make some assumptions on the noisy gradient under stationary input:
\begin{assumption} \label{assumption}
For all period $t=1,\ldots,T$, 
we assume that the noisy gradients satisfy the following.
\begin{enumerate}
    \item   Conditional unbiasedness: $\mathbb{E}[\nabla \varphi(\mathbf{b},\varepsilon_t) \mid \mathcal{H}_{t-1}]=\nabla \Phi(\mathbf{b})$ for all $\mathbf{b} \in K$ almost surely;
    \item  Finite variance: $\mathbb{E}[\Vert \nabla \varphi(\mathbf{b},\varepsilon_t) - \nabla \Phi(\mathbf{b})\Vert_\ast^2  \mid \mathcal{H}_{t-1}] \leq C$ with $C\geq 0$ for all $\mathbf{b} \in K$ almost surely.  
\end{enumerate}
\end{assumption}
In an i.i.d. input,  the noisy gradients can be decomposed into  $\nabla \varphi(\mathbf{b},\varepsilon_t)_{ij}=\nabla \Phi(\mathbf{b})_{ij}+ \varepsilon_{ij,t}$ for any $i$ and $j$. Here, the noise $\varepsilon_{ij,t}$ is a random error that describes the perturbation of observations.
Assumption \ref{assumption}.2 is made to avoid the issue of noisy gradients that are not well-defined due to the presence of logarithms. This assumption is equivalent to assuming that the noises $\varepsilon_{ij,t}$ are conditionally bounded in mean square: $\mathbb{E}[ \varepsilon_{ij,t}^2 \mid \mathcal{H}_{i,t-1}]\leq C_{ij}$ with a constant $C_{ij} \geq 0$ for all $i,j$. 
We do not directly assume that the values $\mathbf{v}_t$ are conditionally unbiased because  the noisy gradient involves the logarithm of value, which leads to the biased noisy gradient. Generally speaking, the conditionally unbiased gradients assumption appears very often and is an important assumption in online learning algorithms. Additionally, it is assumed that the noisy gradients are bounded in mean square: $\mathbb{E}[ \Vert \nabla \varphi(\mathbf{b},\varepsilon_t)\Vert^2_\ast  \mid \mathcal{H}_{t-1}] \leq \tilde{C} $ (for convenience, using our notations as an example), see
\cite{beck2003mirror,nedic2014stochastic,mertikopoulos2019learning} for instance. 


Recall that $\mathbf{b}^\ast$ is the optimal bids vector of $\Phi(\mathbf{b})=\mathbb{E}_{\varepsilon\sim \mathcal{P}}[\varphi(\mathbf{b},\varepsilon)]$ such that $\mathbf{b}^\ast=\arg\min_\mathbf{b}\Phi(\mathbf{b})$, and $\mathbf{p}^\ast$ is  the corresponding equilibrium price vector, the following theorem states the performance of the online proportional response with a stationary input. The proof is in Appendix~\ref{sec:app_c}.
\begin{theorem}\label{SW_regret_stationary}
Under Assumptions \ref{assumption2} and \ref{assumption} , suppose each buyer updates  bid price by the online proportional response in each period $t\geq 2$ and sets initial bid as $b_{ij,1}=\frac{1}{MN}$ for any item $j$. When 
the noise $\varepsilon_t$ is  independently and identically drawn from a stationary distribution $\mathcal{P} \in \Delta (\mathcal{S})$, it holds that the fariness regret
$$\mathcal{F} \leq \log MN,$$
and that the individual buyer's regret
   $$\mathcal{R}_i \leq \frac{\overline{v}}{\underline{p}} \sqrt{2\log MN} \sqrt{T}.$$
\end{theorem} 
In light of Theorem \ref{SW_regret_stationary}, we know that the algorithm achieves a problem-dependent  upper bound in fairness regret that  does not scale up with the number of time horizon $T$. 
This is quite different from the typical results in the online learning algorithm, where the regret is the order of $\mathcal{O}(\log T)$ for a strongly convex function, see \cite{nedic2014stochastic}. Using  Proposition \ref{generic_upper_bound}, we obtain that the individual buyer's regret has an order of $\mathcal{O}(\sqrt{T})$. The proof is in Appendix~\ref{sec:app_c} as well.
 

\begin{corollary}\label{convergencerate_Stationary}
    Under conditions of Theorem~\ref{SW_regret_stationary}, let $\mathbf{\hat{p}}=\frac{1}{T}\sum_{t=1}^T \mathbf{p}_t$, it has that
    $$\mathbb{E} \left[\Vert \mathbf{\hat{p}} - \mathbf{p}^\ast \Vert^2_1 \right] \leq \Phi (\mathbf{\hat{b}}) -\Phi (\mathbf{b}^\ast) \leq \frac{\mathcal{F}}{T}\leq \frac{2\log MN}{T}.$$ 
    Moreover, we have that $\Phi(\mathbf{b}_T)\leq \Phi(\mathbf{b}_{T-1}) \leq \Phi(\mathbf{b}_{T-2}) \leq \ldots \leq  \Phi(\mathbf{b}_{1}) $ such that 
    $$\mathbb{E} \left[ \Vert \mathbf{p}_T - \mathbf{p}^\ast \Vert^2_1 \right] \leq \Phi(\mathbf{b}_T) -\Phi(\mathbf{b}^\ast) \leq \frac{\mathcal{F}}{T}\leq \frac{2\log MN}{T}.$$
\end{corollary}
This corollary shows the convergence of the time-averaged sequence of the prices\footnote{Note that a typical representation in the existing literature on no-regret game-theoretic learning uses $\mathbf{\hat{p}}=\sum_{t=1}^T \eta_t \mathbf{p}_t/\sum_{t=1}^T \eta_t $. Let $\eta_t=1$ for all $t$, we obtain $\mathbf{\hat{p}}=\frac{1}{T}\sum_{t=1}^T \mathbf{p}_t$.} and the last iterate convergence. The convergence of the time-averaged sequence just follows the convexity of $\Phi (\mathbf{b})$ as $\mathbf{\hat{p}}$ is well-defined. The convergence of the last iterate rests on the property that the sequence $\left\lbrace \Phi (\mathbf{b}_1),\ldots \Phi (\mathbf{b}_T)\right\rbrace$ is monotonically decreasing. In \cite{zhou2021robust}, the authors have shown that in a $(\lambda,\beta)$-weighted strongly monotone game, i.e., $\sum_{i \in \mathcal{N}}\lambda_i \langle \nabla \varphi_i(\mathbf{b})-\nabla \varphi_i(\mathbf{b}^\prime),\mathbf{b}_i-\mathbf{b}_i^\prime \rangle \geq  \beta \Vert \mathbf{b}-\mathbf{b}^\prime \Vert^2_2$ and $\lambda_i,\beta > 0$ for all $i$\footnote{In \cite{zhou2021robust},  the standard form is expressed as $\sum_{i \in \mathcal{N}}\lambda_i \langle \nabla \varphi_i(\mathbf{b})-\nabla \varphi_i(\mathbf{b}^\prime),\mathbf{b}_i-\mathbf{b}_i^\prime \rangle \leq  -\beta \Vert \mathbf{b}-\mathbf{b}^\prime \Vert^2_2$. While each player in that context aims to maximize their individual payoff, for the sake of comparison, we assume that each player aims to minimize their cost, leading to a similar inequality.}, the joint action of all players achieves the last-iterate convergence rate at $\mathcal{O}(1/T)$ by an online gradient descent algorithm. However, the SH convex program is not a strongly monotone game\footnote{Here, $\sum_{i \in \mathcal{N}} \langle \nabla \varphi_i(\mathbf{b})-\nabla \varphi_i(\mathbf{b}^\prime),\mathbf{b}_i-\mathbf{b}_i^\prime \rangle = \sum_{i \in \mathcal{N}}\sum_{j \in \mathcal{M}}(b_{ij}-b_{ij}^\prime)\log \frac{p_{j}}{p_j^\prime}=D_h(\mathbf{p},\mathbf{p}^\prime)+D_h(\mathbf{p}^\prime,\mathbf{p})\geq \Vert \mathbf{p}-\mathbf{p}^\prime \Vert^2_2$}. 

\begin{proposition}\label{pp3.2}
The online proportional response leads the bid (price) vector to converge to a limit point and the allocation to converge to a single market equilibrium.
\end{proposition}
We emphasize that in linear Fisher markets, the equilibrium price vector is unique, but there are multiple equilibrium allocations. This proposition implies that the proportional response ensures that \emph{the sequence of bid (price) vectors converges to a limit point at the time when the price vector converges}.

\section{Non-stationary Input}
In this section, we consider non-stationary inputs to provide a more realistic analysis of the performance of online proportional response. The notion of regret towards a fixed comparator $\mathbf{b}^\ast$ is commonly used when considering stationary inputs. However, we aim to demonstrate that measuring regret towards a fixed comparator $\mathbf{b}^\ast$ is still a meaningful and reasonable metric to consider even when the input is non-stationary. We examine three types of non-stationary inputs: \emph{independent input with adversarial corruption}, \emph{ergodic input}, and \emph{periodic input}. These non-stationary input models have been previously studied in the context of online decision-making problems (\cite{lykouris2018stochastic}, \cite{chen2019robust}, \cite{balseiro2022best}). 
 For each type of non-stationary input, we also evaluate the performance of online proportional response in terms of fairness regret and individual buyer's regret.

In the following analysis, we use the total variation distance, denoted as $\Vert \mathcal{P}_1 -\mathcal{P}_2 \Vert_{TV}$,  to measure the non-stationarity of the input data. The total variation distance quantifies the difference between two probability distributions $\mathcal{P}_1$ and $\mathcal{P}_2$. Given that the noise $\varepsilon_t$ is drawn from the distribution $\mathcal{P}_t$ for any period $t=1,\ldots,T$, the stationary distribution $\overline{\mathcal{P}}$ is the time-averaged distribution, i.e., $\overline{\mathcal{P}}=\frac{1}{T}\sum_{t=1}^T \mathcal{P}_t$.


\subsection{Independent Input with Adversarial Corruptions}
Here we consider the case where  $\varepsilon_t$ are independently drawn in each period $t$, but the distribution in each period $t$ can be corrupted (not necessarily identical). Adversarial corruptions typically fall into two primary categories: targeted attacks and non-targeted attacks. Targeted attacks involve the deliberate creation of adversarial corruption with the specific intention of influencing the decision-maker to predict a particular target class, as seen in cases like click fraud. On the other hand, non-targeted attacks aim to manipulate the decision-maker into producing incorrect outputs without a specific target in mind. An example of this is the occurrence of unpredictable surges in demand for certain items, which can lead to misleading or erroneous decisions by the decision maker.

We use total variation distance to measure the level of the adversarial corruptions. With this measure, we assume that the average corruption over $T$ periods is bounded by $\delta$, and the set of distributions over sequences is defined as:
$$\mathcal{C}^{ID}(\delta) :=\left\lbrace \mathcal{P} \in \Delta(\mathcal{S})^T: \frac{1}{T}\sum_{t=1}^T\Vert \mathcal{P}_t-\overline{\mathcal{P}}\Vert_{TV}\leq \delta, \quad   \overline{\mathcal{P}}=\frac{1}{T}\sum_{t=1}^T\mathcal{P}_t \right\rbrace,$$
then $\mathcal{C}^{ID}(\delta)$ is the  set of all independent inputs with mean deviation at most $\delta>0$. The next theorem presents the performance of the online proportional response under independent inputs with adversarial corruptions, and the proofs are available in Appendix \ref{appendix5.1}.
\begin{theorem}\label{nonstationary_independent}
Suppose each buyer  updates  bid price by the online proportional response in each period $t\geq 2$ and sets initial bid as $b_{ij,1}=\frac{1}{MN}$ for any item $j$. 
The noise $\varepsilon_{t}$   is independently drawn from some distributions $\mathcal{P} \in \mathcal{C}^{ID}(\delta)$. Then, it holds for any $T\geq 2$, the fairness regret
$$\mathcal{F} \leq 2\delta T+\log MN,$$ and 
 the individual buyer's regret
$$\mathcal{R}_i \leq \frac{\overline{v}}{\underline{p}} \sqrt{4  \delta T+2\log MN} \sqrt{T}.$$
\end{theorem}
When each buyer employs the online proportional response to update bid prices, the fairness regret upper bound  becomes $\mathcal{O}(\delta T)$, while the upper bounds for the individual buyer's regret is $\mathcal{O}(\sqrt{T}+\sqrt{\delta}T)$. These upper bounds show that the performance of the online proportional response degrades linearly in the average corruption $\delta$. Theorem \ref{nonstationary_independent} indicates that if the   average corruption $\delta$ remains below or equal to  $\mathcal{O}(T^{-1})$, the fairness regret is still at most $\mathcal{O}(1)$, and the individual buyer's regret is  $\mathcal{O}(\sqrt{T})$. Moreover, the theorem establishes a relationship between the stationary (i.i.d.) input and the adversarial input: when the average corruption $\delta=0$, meaning that the noises are drawn from a stationary distribution, the aforementioned upper bounds coincide with those upper bounds in Theorem~\ref{SW_regret_stationary}; the extreme case happens when the average corruption $\delta$ is adversarial, i.e., $\delta$ is a constant, and thus these upper bounds are $\mathcal{O}(T)$.

\subsection{Ergodic Input}
Our attention now turns to investigating ergodic input processes, characterized by a decreasing dependency between data over time. However, it is worth noting that strong correlations may still be evident among data samples obtained from adjacent time periods. The irreducible and aperiodic Markov chains can be considered as an example of the ergodic process. 

Denote a stochastic process by $\mathcal{P}\in \Delta(\mathcal{S})^T$. Let $\varepsilon_{1:t}=(\varepsilon_\tau)_{\tau=1}^t$ be the sequence of inputs up to time $t$ and $\mathcal{P}_t(\varepsilon_{1:\tau})$ be the conditional distribution of $\varepsilon_t$ given $\varepsilon_{1:\tau}$ for $\tau<t$. For every $\kappa \in [T]$, we measure the total variation distance between the distributions in period $t+\kappa$ conditional on the input at the beginning of period $t$ and a one-period distribution $\overline{\mathcal{P}}\in \Delta(\mathcal{S})$ by
$$TV_\kappa(\mathcal{P},\overline{\mathcal{P}})=\sup_{\varepsilon_1,\ldots,\varepsilon_t}\sup_{t=1,\ldots,T-\kappa} \Vert\mathcal{P}_{t+\kappa}(\varepsilon_{1:t-1})-\overline{\mathcal{P}} \Vert_{TV}.$$
Moreover, if $\overline{\mathcal{P}}$ is the stationary distribution of the process $\mathcal{P}$, $TV_\kappa(\mathcal{P},\overline{\mathcal{P}})$ denotes the maximum distance between the $\kappa$-step transition probability and the stationary distribution. Define the set of all stochastic processes with $\kappa$-step distance from stationary as follows: 
$$\mathcal{C}^{E}(\delta,\kappa) :=\left\lbrace \mathcal{P} \in \Delta(\mathcal{S})^T: TV_\kappa(\mathcal{P},\overline{\mathcal{P}})\leq \delta  \mbox{ for some } \overline{\mathcal{P}} \in \Delta(\mathcal{S})\right\rbrace.$$
We make a mild assumption required in the subsequent analysis. 
\begin{assumption}
    There exists a positive constant $\overline{\varphi}$ such that for any $t \geq 1$
    $$ \sup_{\mathbf{b}} \vert \Phi(\mathbf{b}) -\varphi(\mathbf{b},\varepsilon_t) \vert \leq  \overline{\varphi}.    $$
\end{assumption}
The next theorem presents the performance of the algorithm under the ergodic input, and proofs are available in Appendix \ref{appendix5.2}.
\begin{theorem}\label{nonstationary_Ergodic}
Suppose each buyer  updates  bid price by the online proportional response in each period $t\geq 2$ and sets initial bid as $b_{ij,1}=\frac{1}{MN}$ for any item $j$. 
The noise $\varepsilon_t$ is  drawn from an ergodic process $\mathcal{P} \in \mathcal{C}^{E}(\delta,\kappa)$.  Then, it holds for any $T\geq 2$, that the fairness regret 
$$\mathcal{F} \leq 2(T-\kappa)\delta+2\kappa \overline{\varphi}+ \log MN,$$ and
the individual buyer's regret
$$\mathcal{R}_i \leq \frac{\overline{v}}{\underline{p}} \sqrt{4(T-\kappa)\delta+4\kappa \overline{\varphi}+ 2\log MN} \sqrt{T}.$$
\end{theorem}
The convergence rate of the mirror descent for unconstrained stochastic optimization problems with ergodic input was established in \cite{duchi2012ergodic}, where the authors assumed that the gradient of the objective function can be bounded by a constant. However, it should be noted that the online proportional response does not benefit from this assumption. In light of this, we present the convergence rate without relying on the bounded gradient assumption. 
The key step in the proof is that  $\nabla \varphi(\mathbf{b}_t,\varepsilon_{t+\kappa})$ is a nearly unbiased estimate of $\nabla \Phi(\mathbf{b}_t)$.  Specifically, the total variance distance is applied to measure the difference between them. 

Theorem~\ref{nonstationary_Ergodic} leads to an objective gap of order $\mathcal{O}((T-\kappa)\delta+\kappa)$, an individual buyer's regret of order $\mathcal{O}(\sqrt{(T-\kappa)\delta+\kappa} \sqrt{T}$). If $\kappa=0$ and $\delta=0$, the input is stationary and we recover the bound in Theorem~\ref{SW_regret_stationary}. When only $\kappa=0$, it implies that the correlation between the input does not exist but the distributions in each period are not identical. Then we recover the bound in Theorem~\ref{nonstationary_independent}. Now, we consider the input being an irreducible and aperiodic Markov process and let $\mathcal{P}_t(\varepsilon_\tau)$ denote the distribution at period $t$ when the state is $\varepsilon_\tau$, for $\tau<t$. Suppose that the total variation distance $\sup_{\varepsilon} \Vert \mathcal{P}_t(\varepsilon_\tau)-\overline{\mathcal{P}}\Vert_{TV} \leq K\alpha^{t-\tau}$ for some $K>0$ and $\alpha\in (0,1)$, then the $\kappa$-step distance decreases exponentially in $\kappa$. In this case, we obtain an objective gap with an order $\mathcal{O}((T-\kappa)\alpha^{\kappa}+\kappa)$. Setting $\kappa=-\log T/(\log \alpha)$ yields a gap of $\mathcal{O}(\log T)$, an individual buyer's regret bound of $\mathcal{O}(\sqrt{T(\log T)} )$. 

\subsection{Periodic Input}
Another nonstationary input we considered is the periodic input, which means the data have daily, weekly, or seasonal patterns that repeat over time. For example, the frequency of web ad clicks will vary during the day and night, but the pattern tends to be consistent from day to day.

Assume that~$T$ has been divided into total $Q$ partitions and denoted by $I_q=t_{q+1}-t_{q}\geq 1$  the length of partition $q$ and $t_1=1$, $t_{Q+1}=T$. 
Then, each partition $q$ involves the time interval from $t_q$ to $t_{q+1}-1$. The $\varepsilon_t$ within the same partition can be arbitrarily correlated but partitions are identical and independent of each other. 
Let $\mathcal{P}_{\tau:t}$ denotes the joint distributions on period $\tau \leq t$. The total variation distance between the distributions conditional on the data at the beginning of a partition $q$ and a stationary distribution $\overline{\mathcal{P}} \in \Delta(\mathcal{S})$ is 
$$TV(\mathcal{P},\overline{\mathcal{P}},Q )= \sup_{\varepsilon_1,\ldots,\varepsilon_t}\frac{1}{T} \sum_{q=1}^Q \left\Vert \sum_{t=t_q}^{t_{q+1}-1} \mathcal{P}_{t}(\varepsilon_{1:t_q-1})-\overline{\mathcal{P}}\right\Vert_{TV} \leq \delta.$$
If there are $T$ partitions, it means that the noises are identical and independent of each other, then $\delta =0$. If there is only 1 partition, meaning that the noises are arbitrarily correlated, then $\delta$ may be equal to $\mathcal{O}(T)$. 
Formally, we denote by $\mathcal{C}^{P}(\delta,Q)$ the class of all {$Q$-periodic} distributions. The following theorem presents  the performance of the algorithm under the periodic input, and proofs are
available in Appendix \ref{appendix5.3}.
\begin{theorem}\label{nonstationary_Periodic}
Suppose each buyer  updates  bid price by the online proportional response in each period $t\geq 2$ and sets initial bid as $b_{ij,1}=\frac{1}{MN}$ for any item $j$.
The noises $\varepsilon_t$  are drawn from a periodic process $\mathcal{P} \in \mathcal{C}^{P}(\delta,Q)$. Then, it holds for any $T\geq 2$ that the fairness regret
$$\mathcal{F} \leq \log MN+2\delta T,$$
and the individual buyer's regret
$$\mathcal{R}_i \leq \frac{\overline{v}}{\underline{p}} \sqrt{4\delta T+2\log MN} \sqrt{T}.$$
\end{theorem}
The above theorem shows that the fairness regret upper bound under periodic input is $\mathcal{O}(\delta T)$. Thus, if $\delta$ is below or equal to $\mathcal{O}(T^{-1})$, the fairness regret is at most $\mathcal{O}(1)$ and the individual buyer's regret is at most $\mathcal{O}(\sqrt{T})$. Suppose that the  partition is independently and identically distributed, that is, the class of all distributions are
$\mathcal{C}^{P}(\delta,Q)=:\left\lbrace \mathcal{P} \in \Delta(\mathcal{S}^q)^{T/Q}: \mathcal{P}_{t_1:t_2}=\mathcal{P}_{t_2+1:t_3}=\ldots=\mathcal{P}_{t_{Q}+1:t_{Q+1}} \right\rbrace$, and $T$ is separated into $Q$ equal length,  In this case, the total variation distance $TV(\mathcal{P},\overline{\mathcal{P}},Q )=0$, which implies that it recovers the results under the stationary input.

\textbf{Remark.} In comparison to previous works that have considered non-stationary input, such as  \cite{duchi2012ergodic} and \cite{balseiro2022best}, our approach differs in several aspects. While both of these works ensured bounded gradients and employed the mirror descent algorithm with a time-decreasing step-size, the distance between two consecutive decisions, denoted as $\Vert \mathbf{x}_{t+1}-\mathbf{x}_t \Vert$, was proposed to be bounded by $ \leq C \eta_t $. Here, $C$  represents the upper bound of the gradient and $\eta_t$ is the step-size. The results of regret all rely on this proposition. However, in the case of the SH convex program, the gradient is unbounded, and the online proportional response employs a fixed step-size. Therefore, our results of regret significantly deviate from the previous results.

\section{Numerical Experiments}
We illustrate the online proportional response (ALGORITHM~\ref{alg:one}) with numerical examples under various synthetic inputs. The results of our experiments verify the theoretical bounds obtained in Theorem~\ref{SW_regret_stationary}, Theorem~\ref{nonstationary_independent}, Theorem~\ref{nonstationary_Ergodic} and Theorem~\ref{nonstationary_Periodic}. 

\textbf{Stationary inputs}: To numerically assess the efficacy of the online proportional response under i.i.d. inputs, we examine five distinct market instances. For each instance, we conduct 50 random trials and present the average fairness regret. In each trial, the initial value  $\mathbf{v}_{i,0}$ of each buyer $i \in \mathcal{N}$ is drawn from a uniform distribution. Subsequently, each $v_{ij,0}$ undergoes normalization, i.e., $\sum_{j \in \mathcal{M}} v_{ij,0}=1, v_{ij,0}>0$. The noise $\varepsilon_{ij,t}$ is independently and identically generated from a normal distribution with zero mean and a standard deviation $\sigma$. Notably, the value of $\sigma$ diminishes with an increase in market size, taking on values of $0.05,0.04,0.03,0.02,0.01$, respectively. 

\textbf{Non-stationary inputs:} To validate the performance of the online proportional response under non-stationary inputs, we generate noises based on the following input models: \begin{enumerate}
    \item Independent input with adversarial corruptions: for each period $t$, $\varepsilon_t \sim \mathcal{P}_t$, and $\mathcal{P}_t$ is a normal distribution with  mean $\mu_t$ and standard deviation $\sigma$, where $\mu_t$ is uniformly draw from the interval $[-0.01,0.01]$.
    \item Ergodic input: for each period $t$, the noise follows an autoregressive process  $\varepsilon_{t}=\alpha \varepsilon_{t-1}+\beta_t$ with $\alpha=0.6$ and $\beta_t$ is generated i.i.d. from a normal distribution with zero mean and standard deviation $\sigma=0.01$.
    \item Periodic input: The time horizon is equally divided into $50$ partitions,  each consisting of 100 periods. Let $\mathcal{P}_{t_q+1:t_{q+1}}$ be a set of distributions of partition $q$, where $q=1,\ldots,Q$, and each $\mathcal{P}_{t_q+1:t_{q+1}}$ is sampled randomly. The noises are generated from $\mathcal{P}_{t_q+1:t_{q+1}}$ and then subjected to a random permutation over the 
 $100$ sampled noises.
\end{enumerate}
Figure~\ref{fig:whole} illustrates the fairness regret obtained by the online proportional response under both stationary and non-stationary inputs. Each data point represents the average fairness regret across 50 random trials. In Figure~\ref{fig:sub1}, certain instances demonstrate a slight increase in fairness regret over time. This phenomenon is attributed to the myopic nature of the online proportional response, wherein each buyer's bids are susceptible to stochastic noises, resulting in bids that consistently hover around optimal values. Consequently, the expected objective value of the convex program $\Phi (\mathbf{b}_t)$ is consistently lower than the optimal value~$\Phi (\mathbf{b}^\ast)$. In Figure~\ref{fig:sub2}, we use the term ``mild'' to represent ``independent input with adversarial corruptions".  It is evident that both mild inputs and ergodic inputs exhibit a notable increase in fairness regret under i.i.d. inputs when the parameter $\delta$ (as constraining $\delta$ is computationally expensive). However, the fairness regret for periodic inputs remains nearly identical to that of stationary inputs, thereby validating Theorem~\ref{nonstationary_Periodic}.
\begin{figure}[h]
  \begin{subfigure}{0.5\textwidth}
    \centering
    \includegraphics[width=\linewidth]{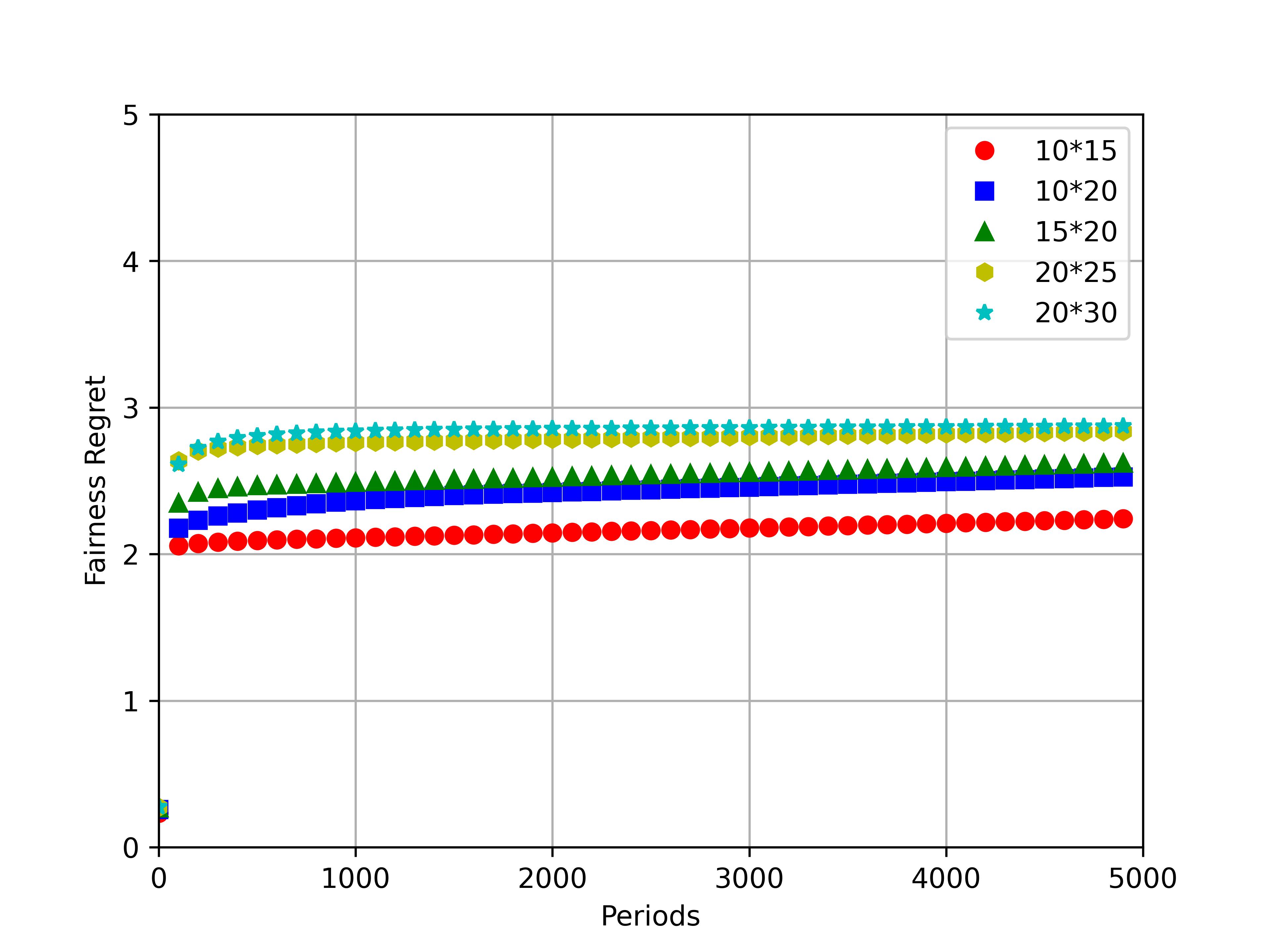}
    \caption{stationary inputs}
    \label{fig:sub1}
  \end{subfigure}%
  \begin{subfigure}{0.5\textwidth}
    \centering
    \includegraphics[width=\linewidth]{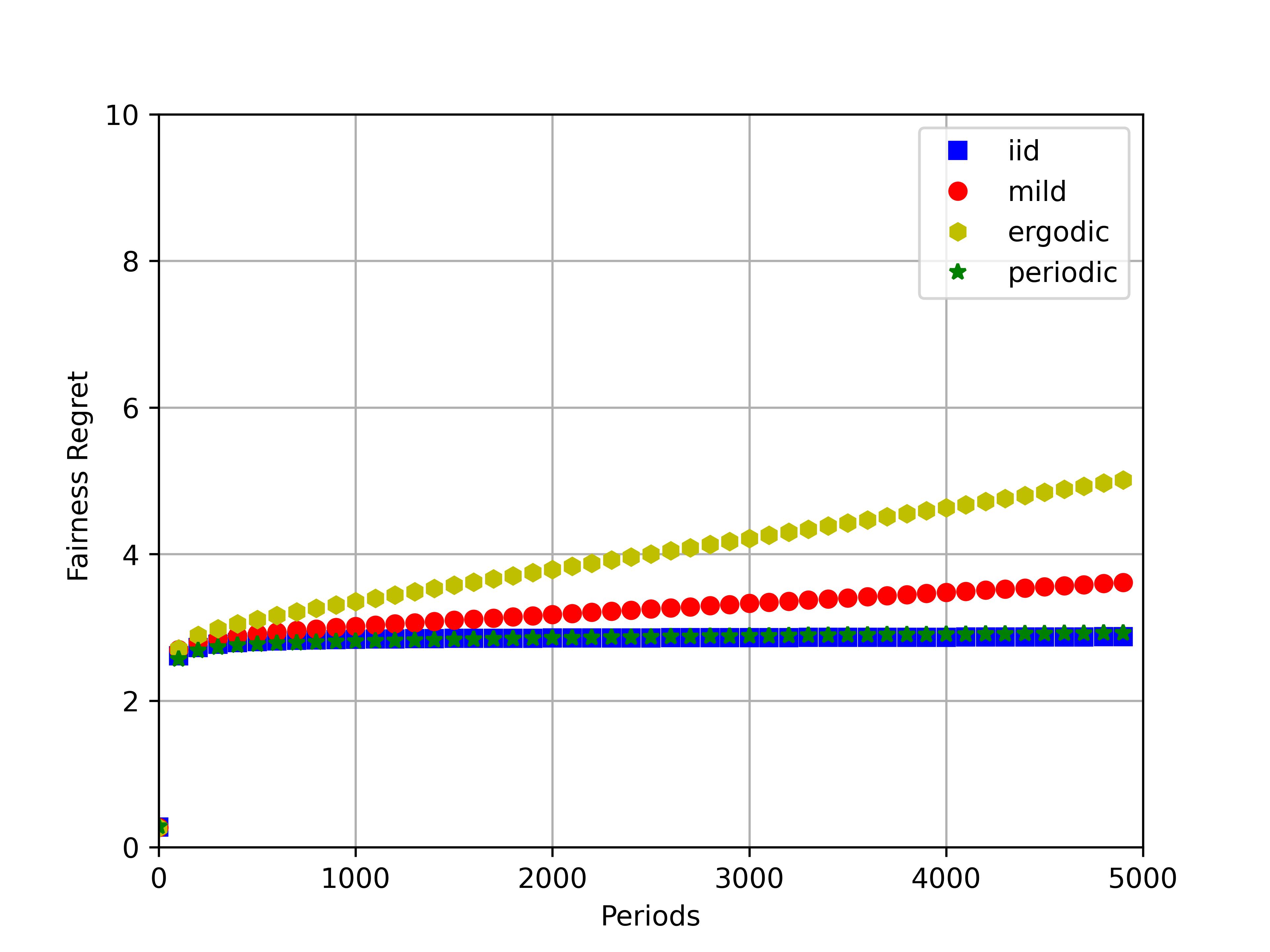}
    \caption{non-stationary inputs ($N=20, M=30$)}
    \label{fig:sub2}
  \end{subfigure}
  \caption{The fairness regret of the online proportional response (ALGORITHM~\ref{alg:one})  under various inputs.}
  \label{fig:whole}
\end{figure}

We also need to validate the individual buyer's regret in our theorems. Given that the expected maximum utility $\mathbb{E}[u_i^\ast]$ varies across trials, we compute the relative individual buyer's regret using the formula: $\mathcal{R}_i/(T\mathbb{E}[u_i^\ast])$. Theorem~\ref{SW_regret_stationary} establishes that the individual buyer's regret 
 $\mathcal{R}_i$ is bounded by $\mathcal{O}(\sqrt{T})$, implying that the logarithm of the relative individual buyer's regret should exhibit a linear relationship with $\log T$ and a slope of approximately~$-0.5$.  We perform a similar analysis for non-stationary inputs, fitting the log of the relative individual buyer's regret against $\log T$. Specifically, we focus on the results for buyer 1 and present them in  Figure~\ref{fig:individual}.  Notably, both fitted linear functions display a slope around -0.5. This consistency arises from  Proposition~\ref{generic_upper_bound}, which establishes that the individual buyer's regret is contingent on the fairness regret. Given that the disparities in fairness regret under various inputs are relatively small, the individual buyer's regret exhibits proximity under different inputs.
\begin{figure}[h]
		\centering
  \includegraphics[width=8cm]{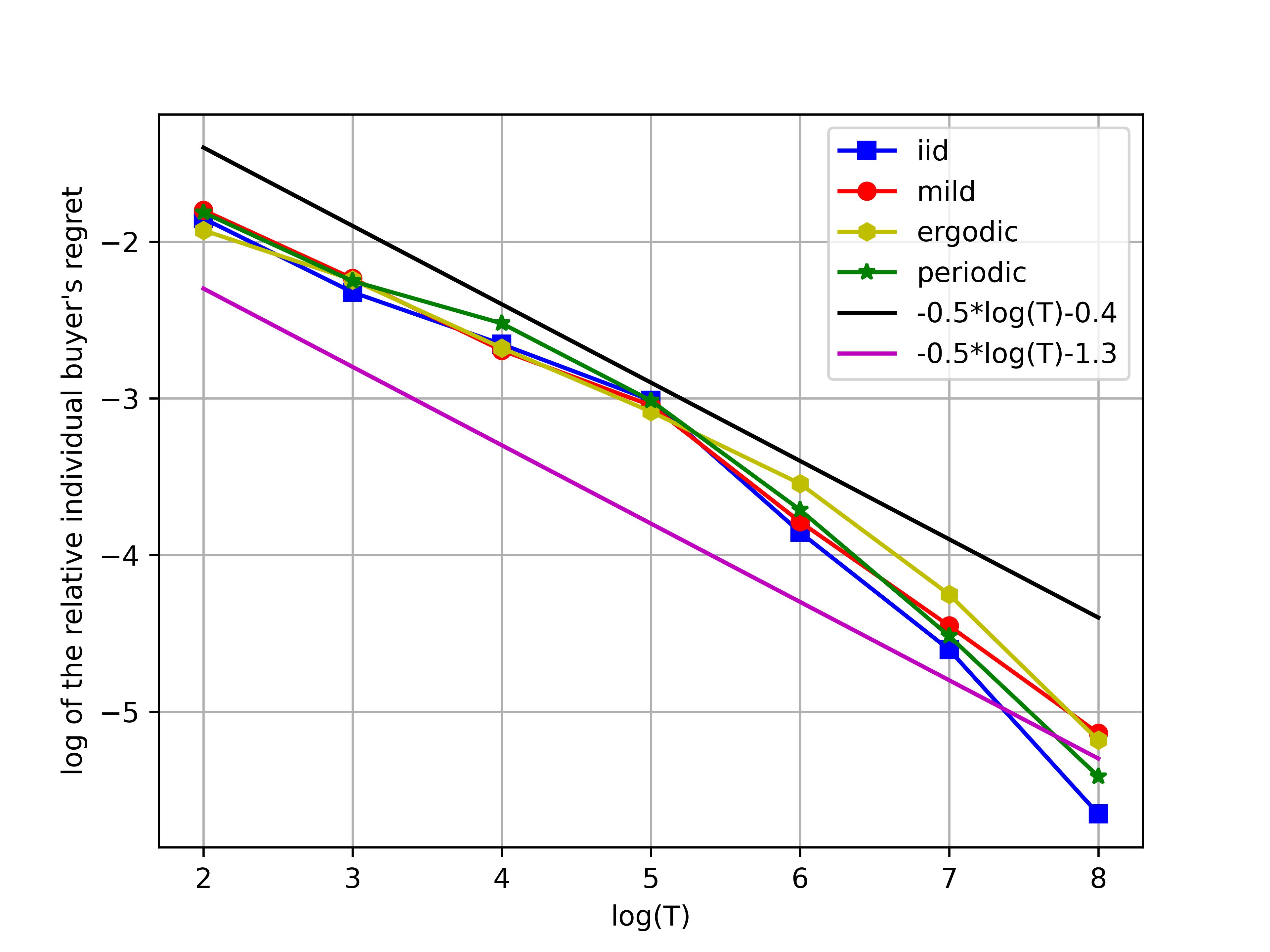}
 	\caption{The log-log plots of individual buyer's regret of the online proportional response under various inputs ($N=20$, $M=30$).}
  \label{fig:individual}
\end{figure}

Furthermore, we conducted an experiment to emphasize the significance of a fixed constant step-size ($\eta_{i,t}=1$ for all $i\in \mathcal{N}$). 
We compared it against two commonly used time-varying step-sizes in the online learning literature: for all $i\in \mathcal{N}$,  $\eta_{i,t}=\frac{1}{t}$ and $\eta_{i,t}=\frac{1}{\sqrt{t}}$. 
The closed-form expression for bid updates with time-varying step-size is:
\begin{equation} \label{stepsize_update}
    b_{ij,t}=B_{i}\frac{b_{ij,t-1}(v_{ij,t}/p_{j,t-1})^{\eta_{i,t}}}{\sum_{j \in \mathcal{M}} b_{ij,t-1}(v_{ij,t}/p_{j,t-1})^{\eta_{i,t}}}, i \in \mathcal{N}, j \in \mathcal{M}.
\end{equation}
As illustrated in Figure~\ref{fig:step}, the fairness regret exhibits a linear increase, indicating that the time-varying step-size hinders the convergence of market prices and allocations to the market equilibrium.
\begin{figure}[h]
		\centering
  \includegraphics[width=8cm]{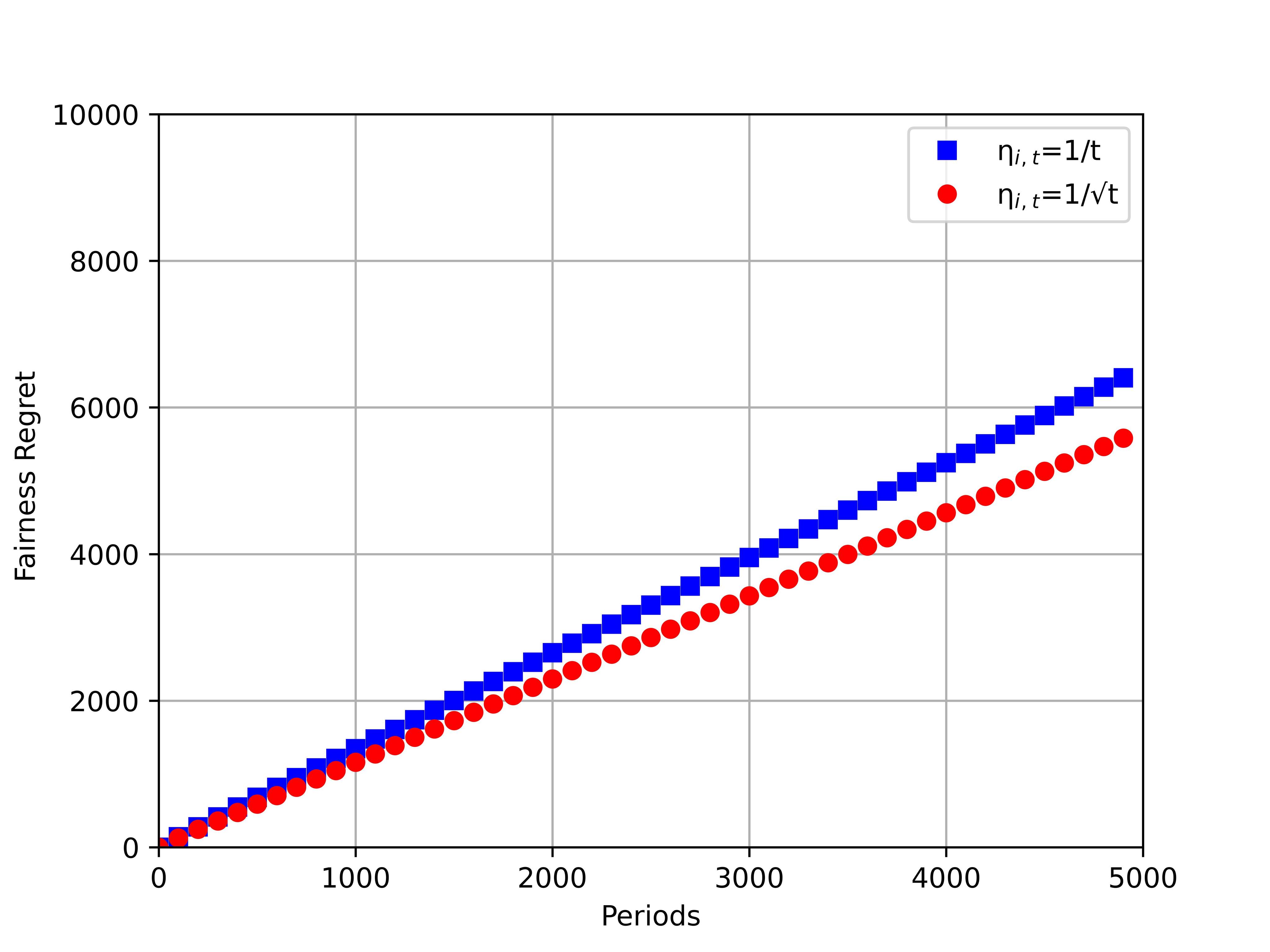}
 	\caption{The performance of bid updates with  time-varying step-size }
  \label{fig:step}
\end{figure}

\section{Conclusions and Future Works}
In this paper, we conducted a comprehensive study on online Fisher markets, characterized by continuous item supply and buyers with varying valuations. Traditional approaches to market equilibrium computation necessitate full information about buyers' valuations and budgets, posing limitations. To overcome this constraint, we introduced the online proportional response, a bidding strategy enabling buyers to adapt bids based on their observed values. Notably, this strategy facilitates continuous market clearance without buyers having to disclose private information, offering advantages for advertisers on platforms like Google and Amazon, which frequently host ad auctions.

We reiterated connections between the online proportional response and online mirror descent for online Fisher markets, analyzing the former's performance under diverse input models describing the environmental state. In the context of a stationary environment, we determined that fairness regret is upper-bounded by a constant, dependent on market size but independent of running periods. Additionally, we showed the last-iterate convergence rate of prices and the asymptotic convergence of bids to a limit point. In the  non-stationary environments, fairness regret is associated with a constant measuring the deviation from a stationary to a non-stationary environment.

Future research avenues include exploring utility functions beyond linear ones, such as Constant Elasticity of Substitution utility functions. Additionally, the setting of divisible items studied in this work may not hold in real-world scenarios where items are indivisible. An engaging avenue for research involves investigating bidding strategies for buyers under such circumstances. Furthermore, given the resemblance of the Fisher market to the Arrow-Debreu model, where buyers exchange items rather than making purchases directly from the market, it is worthwhile to examine the feasibility and relevance of applying online proportional response in the Arrow-Debreu model.

\bibliographystyle{apalike}  
\bibliography{ofm}

\appendix
\section{Preliminaries}
\subsection{Characterizations of the EG Convex Program}
\label{AppendixA.1}
We briefly introduce some basic characterizations of the EG Convex Program. 
The Lagrangian function of the EG convex program (\ref{EG}) is as follows:
$$\mathcal{L}_{EG}=\sum_{i \in \mathcal{N}} B_i\log u_i +\sum_{j \in \mathcal{M}} p_j(1-\sum_{i \in \mathcal{N}}x_{ij}) +\sum_{i \in \mathcal{N}}\sum_{j \in \mathcal{M}}\lambda_{ij}x_{ij},$$
where $p_j$ and $\lambda_{ij}$ are the multipliers for the corresponding constraints. Taking derivatives with respect to $x_{ij}$, it leads to 
$$\frac{\partial \mathcal{L}_{EG}}{\partial x_{ij}}=\frac{B_iv_{ij}}{\sum_{j \in \mathcal{M}} v_{ij}x_{ij}}-p_j+\lambda_{ij}=0, \forall i,j.$$
Manipulating the above equation gives that $$\frac{B_iv_{ij}}{p_j-\lambda_{ij}}=\sum_{j \in \mathcal{M}}v_{ij}x_{ij}=u_i, \forall i,j.$$ 
By the complementary slackness condition, it holds that $p_j(1-\sum_{i \in \mathcal{N}}x_{ij})=0$ and $\lambda_{ij}x_{ij}=0$, then we can derive:\\
\begin{itemize}
    \item If $x_{ij}>0$ and $\lambda_{ij}= 0$, then $u_i=\frac{B_iv_{ij}}{p_j}$.
    \item If $x_{ij}=0$ and $\lambda_{ij}\geq 0$, then $u_i=\frac{B_iv_{ij}}{p_j-\lambda_{ij}} \geq \frac{B_iv_{ij}}{p_j}$.
\end{itemize}
These two relations implies that 
$$\frac{v_{ij}}{p_j} \leq \frac{u_i}{B_i},$$
and the equality holds when $x_{ij}>0$. The Lagrangian function of each buyer's utility maximization problem (\ref{single}) is:
$$\mathcal{L}_i=\sum_{j \in \mathcal{M}}v_{ij}x_{ij}+\beta_i(B_i-\sum_{j \in \mathcal{M}}p_jx_{ij})+\sum_{j \in \mathcal{M}} \lambda_{ij}x_{ij}.$$
We also take derivatives with respect to $x_{ij}$ and obtain that
$$\frac{\partial \mathcal{L}_i}{\partial x_{ij}}=v_{ij}-\beta_ip_j+\lambda_{ij}=0, \forall i,j,$$
where $\beta_i$ is the  multiplier of the budget constraint $\sum_{j \in \mathcal{M}} p_jx_{ij} \leq B_i$, and $\beta_i=\frac{v_{ij}+\lambda_{ij}}{p_j}$. Similarly, by
the complementary slackness condition, it follows:
\begin{itemize}
    \item If $x_{ij}>0$ and $\lambda_{ij}= 0$, then $\beta_i=\frac{v_{ij}}{p_j}$.
    \item If $x_{ij}=0$ and $\lambda_{ij}\geq 0$, then $\beta_i=\frac{v_{ij}+\lambda_{ij}}{p_j}$.
\end{itemize}
The above two relations imply that each buyer $i$ is willing to buy the items $j$ with the utility price $\frac{v_{ij}}{p_j}$. Moreover, each buyer $i$ prefers to buy the items with maximum utility price among all items, that is, buy item $j:= \arg \max \left\lbrace \frac{v_{ij}}{p_j}: \forall j \right\rbrace$. Therefore, the utility of each buyer $i$ is equal to $\beta_iB_i$, where $\beta_i = \max_j \left\lbrace \frac{v_{ij}}{p_j}: \forall j \right\rbrace$.
\subsection{Derivation of the Shmyrev Convex Program}
The dual program of the EG convex program (\ref{EG}): 
 \begin{align}\label{dualEG}
     \min &\sum_{j \in \mathcal{M}} p_j+ \sum_{i \in \mathcal{N}}B_i\log \beta_i ,\\\notag
     \mbox{subject to \,} & p_j\beta_i \geq v_{ij}, \forall i,\forall j,
 \end{align}
 In here, $x_{ij}$ is the dual variable of the constraints $p_j\beta_i \geq v_{ij}, \forall i,j$.  Taking logarithm on both sides of the constraints $p_j \beta_i\geq v_{ij}, \forall i,\forall j$, it leads to 
 $$\log p_j +\log {\beta_i}\geq \log v_{ij}, \forall i,j.$$
 Let $q_j=\log p_j$ and $\delta_i=\log \beta_i$ , hence the dual program of EG becomes:
  \begin{align}
     \min &\sum_{j \in \mathcal{M}} e^{q_j}+ \sum_{i \in \mathcal{N}}B_i\delta_i, \\\notag
    \mbox{subject to} & \; \delta_i+q_j \geq \log v_{ij}, \forall i,\forall j.
 \end{align}
 Now we write down the dual program of (\ref{dualEG}), which is called as Shymrev program. 
 We transform the dual of this program in two steps: first, we write down the conjugate duality of $e^{q_j}$ which is $p_j\log {p_j}-p_j$.
 Let $f(x)=ce^{x}$, then it has $\nabla f(x)=ce^{x}=y$, therefore $x=\log \frac{y}{c}$, the conjugate of $f(x)$ is $f^\ast (y)=xy-f(x)=y\log \frac{y}{c}-ce^{x}=y\log \frac{y}{c}-y$, it also has $\nabla f^\ast (y) = \log y+1-\log c-1=\log \frac{y}{c}=x$. The rest could be transformed by the techniques of linear duality. 
 The Shmyrev program is:
  \begin{align}\label{shmyrev}
     \max &\sum_{i \in \mathcal{N}}\sum_{j \in \mathcal{M}} b_{ij} \log v_{ij}- \sum_{j \in \mathcal{M}}p_j\log p_j \\\notag
     \mbox{subject to} & \sum_{i \in \mathcal{N}}b_{ij}=p_j, \forall j \in \mathcal{M},\\\notag
     & \sum_{j \in \mathcal{M}}b_{ij}=B_i, \forall i \in \mathcal{N}. \notag
 \end{align}
We use 
$\psi(\mathbf{b})=-\sum_{i \in \mathcal{N}} B_i\log u_i= -\sum_{i \in \mathcal{N}} B_i\log \left(\sum_{j \in \mathcal{M}}\frac{{v_{ij}b_{ij}}}{p_j}\right)$ 
such the objective of (\ref{EG}). By Lemma 19 in \cite{birnbaum2011distributed}, it shows that $$\psi(\mathbf{b}) \leq \varphi(\mathbf{b})-B_i\log B_i \quad \mbox{and} \quad \psi(\mathbf{b}^{\ast}) \leq \varphi(\mathbf{b}^{\ast})-B_i\log B_i,$$
such that 
\begin{equation} \label{EGlessSH}
    \psi(\mathbf{b})-\psi(\mathbf{b}^\ast) \leq \varphi(\mathbf{b})-\varphi(\mathbf{b}^\ast).
\end{equation}
\section{Proofs in Section 3}
\subsection{Proof of Lemma \ref{lemma3.1}}
\begin{proof}
Recall that  $D_h(\mathbf{p}, \mathbf{q})=\sum_{j \in \mathcal{M}}p_{j}\log(\frac{p_j}{q_j})$, then we have
    \begin{align*}
    &\varphi(\mathbf{b}_t,\varepsilon_t)-\ell_\varphi(\mathbf{b}_t,\mathbf{b}_{t-1},\varepsilon_t)\\=& \varphi(\mathbf{b}_t,\varepsilon_t)- \varphi(\mathbf{b}_{t-1},\varepsilon_t)-\langle\nabla \varphi(\mathbf{b}_{t-1},\varepsilon_t),\mathbf{b}_t-\mathbf{b}_{t-1}\rangle\\
    =&-\sum_{i \in \mathcal{N}}\sum_{j \in \mathcal{M}}b_{ij,t}\log\left(\frac{v_{ij,t}}{p_{j,t}}\right)+\sum_{i \in \mathcal{N}}\sum_{j \in \mathcal{M}}b_{ij,t-1}\log\left(\frac{v_{ij,t}}{p_{j,t-1}}\right)-\sum_{i \in \mathcal{N}}\sum_{j \in \mathcal{M}}\left(1-\log\left(\frac{v_{ij,t}}{p_{j,t-1}}\right)\right)(b_{ij,t}-b_{ij,t-1})\\
    =&-\sum_{i \in \mathcal{N}}\sum_{j \in \mathcal{M}}b_{ij,t}\log\left(\frac{v_{ij,t}}{p_{j,t}}\right)+\sum_{i \in \mathcal{N}}\sum_{j \in \mathcal{M}}b_{ij,t-1}\log\left(\frac{v_{ij,t}}{p_{j,t-1}}\right)+\sum_{i \in \mathcal{N}}\sum_{j \in \mathcal{M}}b_{ij,t}\log\left(\frac{v_{ij,t}}{p_{j,t-1}}\right)-\sum_{i \in \mathcal{N}}\sum_{j \in \mathcal{M}}b_{ij,t-1}\log\left(\frac{v_{ij,t}}{p_{j,t-1}}\right)\\
    &\;\;\;-\sum_{i \in \mathcal{N}}\sum_{j \in \mathcal{M}}b_{ij,t}+\sum_{i \in \mathcal{N}}\sum_{j \in \mathcal{M}}b_{ij,t-1}\\
    =&\sum_{i \in \mathcal{N}}\sum_{j \in \mathcal{M}}b_{ij,t}\log\left(\frac{p_{j,t}}{p_{j,t-1}}\right)=D_h(\mathbf{p}_t, \mathbf{p}_{t-1}).
\end{align*}
The first equation uses the definition of $\ell_\varphi(\mathbf{b}_t,\mathbf{b}_{t-1},\varepsilon_t)$. The fourth equation uses $\sum_{i \in \mathcal{N}}\sum_{j \in \mathcal{M}}b_{ij,t}=\sum_{i \in \mathcal{N}}\sum_{j \in \mathcal{M}}b_{ij,t-1}$. 
\end{proof}


\subsection{Proof of Lemma \ref{lemma3.2}}
\begin{proof}
The update rule of $\mathbf{b}_{i,t}$ is
$$\mathbf{b}_{i,t}=\argmin\limits_{\sum_{j \in \mathcal{M}} b_{ij}=B_i,b_{ij,t}\geq 0} \left\lbrace \eta_{i,t}\langle  \nabla \varphi(\mathbf{b}_{i,t-1},\varepsilon_t) ,\mathbf{b}_{i}-\mathbf{b}_{i,t-1}\rangle+\sum_{j \in \mathcal{M}} b_{ij}\log \left(\frac{b_{ij,t}}{b_{ij,t-1}}\right)\right\rbrace.$$

 The Lagrangian function is
$$\mathcal{L}_{i,t}=\eta_{i,t} \langle \nabla \varphi(b_{i,t-1},\varepsilon_t), \mathbf{b}_{i}-\mathbf{b}_{i,t-1} \rangle+\sum_{j \in \mathcal{M}}b_{ij,t}\log \left(\frac{b_{ij,t}}{b_{ij,t-1}}\right)+\lambda_{i,t}(\sum_{j \in \mathcal{M}} b_{ij,t}-B_i)+\mu_{ij,t}b_{ij,t}.$$
By the complementary slackness condition, it holds that $\lambda_{i,t}(\sum_{j \in \mathcal{M}} b_{ij,t}-B_i)=0$ and $\mu_{ij,t}b_{ij,t}=0$ for $i \in \mathcal{N},j\in \mathcal{M}$. Since $\sum_{j \in \mathcal{M}} b_{ij,t}=B_i$, so $\lambda_{i,t} \geq 0$, and  $b_{ij,t}$ must be greater than 0  due to the logarithm function, so $\mu_{ij,t}=0$ for  $j\in \mathcal{M}$. Recall that $\nabla \varphi(b_{t-1},\varepsilon_t)_{ij}=1-\log (v_{ij,t}/p_{j,t-1})$, taking derivatives with respect to $b_{ij,t}$, then the first order derivative condition of the  Lagrangian function is 
$$\eta_{i,t}-\eta_{i,t}\log (v_{ij,t}/p_{j,t-1})+\log b_{ij,t}+1-\log b_{ij,t-1}+\lambda_{i,t}=0, $$
therefore, it has that  
$b_{ij,t}=\frac{b_{ij,t-1}}{e^{\lambda_{i,t}+\eta_{i,t}+1}} \left(\frac{v_{ij,t}}{p_{j,t-1}}\right)^{\eta_{i,t}}$. 
Since $\sum_{j \in \mathcal{M}} b_{ij,t}={B_{i}}$, then 
${B_{i}}=\sum_{j \in \mathcal{M}} \frac{b_{ij,t-1}}{e^{\lambda_{i,t}}+\eta_{i,t}+1} \left(\frac{v_{ij,t}}{p_{j,t-1}}\right)^{\eta_{i,t}}$, it leads to 
$$b_{ij,t}={B_{i}}\frac{b_{ij,t-1}(v_{ij,t}/p_{j,t-1})^{\eta_{i,t}}}{\sum_{j \in \mathcal{M}} b_{ij,t-1}(v_{ij,t}/p_{j,t-1})^{\eta_{i,t}}}.$$
Set $\eta_{i,t}=1$ for all $i,t$, we conclude this lemma.
\end{proof}
\subsection{Proof of Lemma \ref{lemma3.3} }
Taking expectations on both sides of (\ref{EGlessSH}), we obtain the results.

\subsection{Proof of Proposition \ref{generic_upper_bound}}
From Lemma \ref{lemma3.3}, we have that 
$$\Psi (\mathbf{b}_t) -\Psi (\mathbf{b}^\ast) \leq \Phi (\mathbf{b}_t) -\Phi (\mathbf{b}^\ast) $$
From Appendix \ref{AppendixA.1}, at the equilibrium price $\mathbf{p}^\ast$, each buyer purchases the item with the maximum bang-per-buck ratio.  Define $\beta_i^\ast =\max_j \left\lbrace \frac{\mathbb{E}[v_{ij}]}{p_j} :\forall j\right\rbrace$, then the maximum expected utility of each buyer in each period $\mathbb{E}[u_i^\ast]=\beta_i^\ast B_{i}$. We derive that
$$\Psi (\mathbf{b}) -\Psi (\mathbf{b}^\ast)=\mathbb{E}\left[\sum_{i \in \mathcal{N}} B_{i} \log \left(\frac{u_i^\ast}{u_{i,t}}\right)\right]=\mathbb{E}\left[\sum_{i \in \mathcal{N}} \frac{u_i^\ast}{\beta_i^{\ast}} \log \left(\frac{u_i^\ast}{u_{i,t}}\right)\right] \geq \mathbb{E}\left[\frac{\underline{p}}{\overline{v}} \sum_{i \in \mathcal{N}} u_i^\ast\log \left(\frac{u_i^\ast}{u_{i,t}}\right)\right]=\mathbb{E}\left[\frac{\underline{p}}{\overline{v}} D_h(\mathbf{u}^{\ast},\mathbf{u}_t))\right].$$
We obtain that 
$$\mathbb{E}\left[\frac{\underline{p}}{\overline{v}} D_h(\mathbf{u}^{\ast},\mathbf{u}_t))\right] \leq \Psi (\mathbf{b}_t) -\Psi (\mathbf{b}^\ast).$$
Note that $u_i^\ast \leq \overline{u}_i$ for any $i$, then $h(u_i)=\sum_{i \in \mathcal{N}} (u_i\log u_i-u_i)$ is $1/(\sum_{i \in \mathcal{N}} \overline{u}_i)$-strongly convex with respect to the $\Vert \cdot \Vert_1$ norm over $u_i >0$ (see Lemma 2 in \cite{balseiro2022best}). It has that
$$\mathbb{E}\left[\frac{\underline{p}}{2\overline{v}(\sum_{i \in \mathcal{N}} \overline{u}_i)}\Vert\mathbf{u}^{\ast}-\mathbf{u}_t \Vert^2 \right] \leq \mathbb{E}\left[ \frac{\underline{p}}{\overline{v}} D_h(\mathbf{u}^{\ast},\mathbf{u}_t) \right] \leq \Psi (\mathbf{b}_t) -\Psi (\mathbf{b}^\ast) \leq \Phi (\mathbf{b}_t) -\Phi (\mathbf{b}^\ast),$$
since $\overline{u}_i=(B_i\overline{v}_i)/\underline{p}_j$ and $\sum_{i \in \mathcal{N}}B_i=1$, it turns out $\sum_{i \in \mathcal{N}} \overline{u}_i \leq \sum_{i \in \mathcal{N}}(B_i\overline{v})/\underline{p} =\overline{v}/\underline{p}$, then
$$ \mathbb{E} \left[u_i^\ast -u_{i,t}\right]\leq \frac{\sqrt{2}\overline{v}}{\underline{p}} \sqrt{ \Phi (\mathbf{b}_t) -\Phi (\mathbf{b}^\ast) }.$$


\section{Proofs in Section 4} \label{sec:app_c}

\subsection{Proof of Theorem 4.1}
We start to analyze the gap between the objective value of (\ref{SHHindsight}) obtained from $\mathbf{b}_t$ and $\mathbf{b}^\ast$ by recalling Proposition~\ref{adverserial_regret}, it has that
\begin{align*}
    \varphi(\mathbf{b}_t,\varepsilon_t)&\leq \ell_{\varphi}(\mathbf{b}_t,\mathbf{b}_{t-1},\varepsilon_t)+D_h(\mathbf{b}_t, \mathbf{b}_{t-1}) \\
    &\leq \varphi(\mathbf{b}_{t-1},\varepsilon_t)+\langle  \nabla \varphi(\mathbf{b}_{t-1},\varepsilon_t),\mathbf{b}^\ast - \mathbf{b}_{t-1}\rangle+D_h(\mathbf{b}^\ast,\mathbf{b}_{t-1})-D_h(\mathbf{b}^\ast, \mathbf{b}_{t}).
\end{align*}
Since $\mathcal{H}_t$ is the history of the algorithm,  by taking the expectation conditioned on $\mathcal{H}_{t-1}$, since $\mathbf{b}_{t-1} \in \mathcal{H}_{t-1}$, we have that
\begin{align*}
&\mathbb{E}\left[\varphi(\mathbf{b}_t,\varepsilon_t)\vert \mathcal{H}_{t-1} \right]\\
   &\leq \mathbb{E}\left[\varphi(\mathbf{b}_{t-1},\varepsilon_t)\vert \mathcal{H}_{t-1} \right]+ \langle  \mathbb{E}\left[\nabla \varphi(\mathbf{b}_{t-1},\varepsilon_t)\vert \mathcal{H}_{t-1} \right], \mathbf{b}^\ast - \mathbf{b}_{t-1}\rangle+ \mathbb{E}\left[ D_h(\mathbf{b}^\ast,\mathbf{b}_{t-1})-D_h(\mathbf{b}^\ast, \mathbf{b}_{t}) \vert \mathcal{H}_{t-1}\right].
\end{align*}
Construct a process $Z_t=\sum_{\tau=1}^t\langle  \nabla \varphi(\mathbf{b}_{\tau-1},\varepsilon_\tau),\mathbf{b}^\ast - \mathbf{b}_{\tau-1} \rangle-\langle  \mathbb{E}\left[\nabla \varphi(\mathbf{b}_{\tau-1},\varepsilon_\tau)\vert \mathcal{H}_{\tau-1} \right], \mathbf{b}^\ast - \mathbf{b}_{\tau-1}\rangle$, which is a martingale with respect to $\mathcal{H}_{t-1}$, the  Optional Stopping Theorem implies that $\mathbb{E}[Z_T]=0$. Using a similar martingale argument for $\varphi(\mathbf{b}_t,\varepsilon_t)$ , we could have that 
\begin{equation}\label{func_diff}
    \mathbb{E} [\varphi(\mathbf{b}_t,\varepsilon_t)-\varphi(\mathbf{b}^\ast,\varepsilon_t) ]\leq \mathbb{E}[D_h(\mathbf{b}^\ast,\mathbf{b}_{t-1})]-\mathbb{E}[D_h(\mathbf{b}^\ast, \mathbf{b}_{t})].
\end{equation}
Summing over $t$, we obtain that 
$$\sum_{t=1}^T \mathbb{E} [\varphi(\mathbf{b}_t,\varepsilon_t)-\varphi(\mathbf{b}^\ast,\varepsilon_t) ] \leq D_h(\mathbf{b}^\ast,\mathbf{b}_1),$$
thus, we could say $\mathbb{E}[\mathcal{F}^\pi] \leq D_h(\mathbf{b}^\ast,\mathbf{b}_1).$\\

\textbf{Individual's regret}: 
Using the convexity of $\Phi$, we obtain that 
$$T\mathbb{E}\left[\Phi(\mathbf{\hat{b}})-\Phi(\mathbf{b}^\ast) \right]\leq \mathbb{E}\left[\sum_{t=1}^T\Phi(\mathbf{b}_t)-\sum_{t=1}^T\Phi(\mathbf{b}^\ast) \right]\leq \log MN,$$
where $\mathbf{\hat{b}} =\frac{1}{T}\sum_{t=1}^T \mathbf{b}_t. $ Consequently, 
$$\mathbb{E}\left[\Phi(\mathbf{\hat{b}})-\Phi(\mathbf{b}^\ast) \right] \leq \frac{\log MN}{T}.$$
According to Proposition~\ref{generic_upper_bound}, the regret for any buyer~$i$ is 
\begin{align*} \label{regret_nonstationary_ID}
   \sum_{t=1}^T \mathbb{E}[ u_i^{\ast}-u_{i,t} ]&= T\mathbb{E}[ u_i^{\ast}- \hat{u}_i ]\\
   &\leq \frac{\sqrt{2}\overline{v}T}{\underline{p}}\mathbb{E}\left[\sqrt{\Phi({\mathbf{\hat{b}}})-\Phi(\mathbf{b}^\ast)}\right] \\
   &\leq \frac{\sqrt{2}\overline{v}T}{\underline{p}}\sqrt{ \frac{\log MN}{T}}\\
   &=  \frac{\overline{v}}{\underline{p}} \sqrt{ 2\log MN} \sqrt{T}.
\end{align*}

\subsection{Proof of Corollary \ref{convergencerate_Stationary}}
\begin{proof}
\textbf{Time-averaged sequence convergence}:
Let $\mathbf{\hat{p}}=\frac{1}{T}\sum_{t=1}^T \mathbf{p}_t$, then $$\hat{p}_j =\frac{1}{T}\sum_{t=1}^T p_{j,t}=\frac{1}{T}\sum_{t=1}^T \sum_{i\in \mathcal{N}} b_{ij,t}=  \sum_{i\in \mathcal{N}} \hat{b}_{ij}, $$
it leads to
$$\mathbb{E}\left[ \Vert\mathbf{\hat{p}}-\mathbf{p}^\ast \Vert^2\right] \leq 2 D_h(\mathbf{\hat{p}},\mathbf{p}^\ast) \leq 2\mathbb{E}\left[\Phi(\mathbf{\hat{b}})-\Phi(\mathbf{b}^\ast) \right] \leq \frac{2 \log MN}{T}. $$
    \textbf{Last-iterate convergence}: 
At first,  we have that 
$$\varphi(\mathbf{b})=\varphi(\mathbf{b}_{t})+\langle \nabla \varphi(\mathbf{b}_{t}),\mathbf{b} - \mathbf{b}_{t}  \rangle+ D_h(\mathbf{p}, \mathbf{p}_t),$$
thus $\varphi(\mathbf{b})$ is 1-\emph{strongly convex} w.r.t $D_h(\mathbf{p}, \mathbf{p}_t)$. Under the stationary input, taking expectation on both sides and replacing $\mathbf{b}^\ast$ into $\mathbf{b}$ such that $\langle \mathbb{E}[\nabla  \varphi(\mathbf{b}^\ast,\varepsilon_t)], \mathbf{b}_t-\mathbf{b}^\ast \rangle \geq 0$, and $D_h(\mathbf{p}, \mathbf{p}_t)\geq \frac{1}{2}\Vert \mathbf{p}-\mathbf{p}_t \Vert^2$, we obtain that 
\begin{equation}\label{stronglyconvex_stationary}
    \mathbb{E}[\varphi(\mathbf{b}_t,\varepsilon_t)]-\mathbb{E}[\varphi(\mathbf{b}^\ast,\varepsilon_t)] =\langle \mathbb{E}[\nabla  \varphi(\mathbf{b}^\ast,\varepsilon_t)], \mathbf{b}_t-\mathbf{b}^\ast \rangle+ \mathbb{E}[D_h(\mathbf{p}^\ast, \mathbf{p}_t)]\geq  \frac{1}{2} \mathbb{E}[\Vert \mathbf{p}^\ast-\mathbf{p}_t \Vert^2].
\end{equation}

Since $\mathbf{b}_t$ is the minimum of (\ref{OMDupdate}), then
$$\varphi(\mathbf{b}_t,\varepsilon_t) \leq \ell_{\varphi}(\mathbf{b}_t,\mathbf{b}_{t-1},\varepsilon_t)+D_h(\mathbf{b}_t, \mathbf{b}_{t-1})\leq \ell_{\varphi}[(\mathbf{b}_{t-1};\mathbf{b}_{t-1}),\varepsilon_t]+D_h(\mathbf{b}_{t-1}, \mathbf{b}_{t-1})=\varphi(\mathbf{b}_{t-1},\varepsilon_t), $$
 taking expectations on both sides, we obtain that 
\begin{equation}\label{objective_time_relation}
\mathbb{E}[\varphi(\mathbf{b}_t,\varepsilon_t)] \leq \mathbb{E}[\varphi(\mathbf{b}_{t-1},\varepsilon_t)]=\mathbb{E}[\varphi(\mathbf{b}_{t-1},\varepsilon_{t-1})].
\end{equation}
Therefore, we have that 
$$\frac{T}{2} \mathbb{E}[\Vert \mathbf{p}^\ast-\mathbf{p}_T \Vert^2]\leq T(\mathbb{E}[\varphi(\mathbf{b}_T,\varepsilon_T)] -\mathbb{E}[\varphi(\mathbf{b}^\ast,\varepsilon_T)])\leq \sum_{t=1}^T \mathbb{E} [\varphi(\mathbf{b}_t,\varepsilon_t)-\varphi(\mathbf{b}^\ast,\varepsilon_t) ] \leq D_h(\mathbf{b}^\ast,\mathbf{b}_1),$$
thus, we obtain the last-iterate convergence rate of the price
$$ \mathbb{E}[\Vert \mathbf{p}^\ast-\mathbf{p}_T \Vert^2] \leq \frac{2D_h(\mathbf{b}^\ast,\mathbf{b}_1)}{T} \leq \frac{2\log MN}{T}, $$
where the last inequality comes from the Lemma \ref{D(b,b1)}. 
\end{proof}
\subsection{Proof of Proposition \ref{pp3.2}}
\begin{proof}
At first, we characterize the convergence behavior of the bid price $\mathbf{b}_t$ generated by the proportional response. Recall that the bid price $\mathbf{b}^\ast$ makes the equilibrium utilities $u_i^\ast$ for any buyer $i$, then an market equilibrium is $(\mathbf{x}^\ast,\mathbf{p}^\ast)$. For any $x_{ij}>0$, it holds that $\frac{v_{ij}x_{ij}^\ast}{u_i^\ast}=\frac{b_{ij}^\ast}{B_i}$, and the online proportional response has that $b_{ij,t}={B_{i}}\frac{v_{ij,t}x_{ij,t-1}}{u_{i,t-1}}$, then
\begin{align} \label{converge_behavior}
   \mathbb{E}\left[ \sum_{i \in \mathcal{N}}\sum_{j \in \mathcal{M}}b_{ij}^\ast \log \frac{b_{ij}^\ast}{b_{ij,t}}\right]&= \mathbb{E}\left[  \sum_{i \in \mathcal{N}}\sum_{j \in \mathcal{M}} b_{ij}^\ast \log \frac{b_{ij}^\ast u_{i,t-1} }{{B_{i}}v_{ij,t}x_{ij,t-1}} \right]\notag \\
    &=\mathbb{E}\left[ \sum_{i \in \mathcal{N}}\sum_{j \in \mathcal{M}} b_{ij}^\ast \log \frac{v_{ij}x_{ij}^\ast u_{i,t-1}}{v_{ij,t}x_{ij,t-1}u_i^\ast} \right]\notag\\
    &=\mathbb{E}\left[ \sum_{i \in \mathcal{N}}\sum_{j \in \mathcal{M}} b_{ij}^\ast \log \frac{b_{ij}^\ast}{b_{ij,t-1}}\right] + \mathbb{E}\left[ \sum_{i \in \mathcal{N}}\sum_{j \in \mathcal{M}} b_{ij}^\ast \log \frac{v_{ij}}{v_{ij,t}}\right]\notag\\
    &-\mathbb{E}\left[ \sum_{i \in \mathcal{N}}\sum_{j \in \mathcal{M}} b_{ij}^\ast \log\frac{p_j^\ast}{p_{j,t-1}}\right]-\mathbb{E}\left[\sum_{i \in \mathcal{N}}\sum_{j \in \mathcal{M}} b_{ij}^\ast \log \frac{u_i^\ast}{u_{i,t-1}} \right]\notag\\
    &=\mathbb{E}\left[\sum_{i \in \mathcal{N}}\sum_{j \in \mathcal{M}} b_{ij}^\ast \log \frac{b_{ij}^\ast}{b_{ij,t-1}}\right]-\mathbb{E}\left[\sum_{j \in \mathcal{M}} p_{j}^\ast \log\frac{p_j^\ast}{p_{j,t-1}}\right]-\mathbb{E}\left[\sum_{i \in \mathcal{N}}\sum_{j \in \mathcal{M}} b_{ij}^\ast \log \frac{u_i^\ast}{u_{i,t-1}}\right].
\end{align}
According to Assumption~\ref{assumption}.1, it has that $\mathbb{E}[\nabla\varphi(\mathbf{b},\varepsilon_t) \vert \mathcal{H}_{t-1}] =\nabla \Phi(\mathbf{b})$, it leads to $\mathbb{E}[1-\log (v_{ij,t}/p_j) \vert \mathcal{H}_{t-1}]=1-\log (v_{ij}/p_j)$, taking expectation on both sides, we get $\mathbb{E}[\log v_{ij,t}]=\mathbb{E} [\log v_{ij}]$. Therefore,  we obtain that
 $\mathbb{E}\left[ \sum_{i \in \mathcal{N}}\sum_{j \in \mathcal{M}} b_{ij}^\ast \log \frac{v_{ij}}{v_{ij,t}}\right] =0 $.
Since the KL divergence $\sum_{j \in \mathcal{M}} p_{j}^\ast \log\frac{p_j^\ast}{p_{j,t-1}}$ is nonnegative and $u_i^\ast \geq u_{i,t-1}$, we conclude that the bid price generated by the proportional response will converge because of
$$\sum_{i \in \mathcal{N}}\sum_{j \in \mathcal{M}}b_{ij}^\ast \log \frac{b_{ij}^\ast}{b_{ij,t}} \leq \sum_{i \in \mathcal{N}}\sum_{j \in \mathcal{M}} b_{ij}^\ast \log \frac{b_{ij}^\ast}{b_{ij,t-1}},$$
and the equality holds if and only if $\sum_{j \in \mathcal{M}} p_{j}^\ast \log\frac{p_j^\ast}{p_{j,t-1}}=0$ and $\sum_{i \in \mathcal{N}}\sum_{j \in \mathcal{M}} b_{ij}^\ast \log \frac{u_i^\ast}{u_{i,t-1}}=0$. In other words, if price and utilities converges, then the bid price also converges.

Now we argue that the allocation will converge to an equilibrium allocation.
Suppose that the bid price converges to a limit point $\hat{\mathbf{b}}$ that gives the equilibrium utilities such that the allocation corresponding to $\hat{\mathbf{b}}$ is an equilibrium allocation $\hat{\mathbf{x}}$, and the corresponding price is $\hat{\mathbf{p}}$. Assume that $\hat{\mathbf{x}}=\mathbf{x}^\ast$, thus we have 
$$\frac{b_{ij}^\ast}{\hat{b}_{ij}}=\frac{x_{ij}^\ast p_j^\ast}{\hat{x}_{ij}\hat{p_j}}=\frac{p_j^\ast}{\hat{p_j}}.$$
For an item $j$, it has 
$$\sum_{i \in \mathcal{N}} b_{ij}^\ast \log \frac{b_{ij}^\ast}{\hat{b}_{ij}}=\sum_{i \in \mathcal{N}} b_{ij}^\ast \log \frac{p_{j}^\ast}{\hat{p}_{j}}= p_{j}^\ast \log \frac{p_{j}^\ast}{\hat{p}_{j}}.$$
 For the sake of contradiction, suppose that the bid price may converge to another limit point $\tilde{\mathbf{b}}\neq \hat{\mathbf{b}}$ and the corresponding allocation is $\tilde{\mathbf{x}} \neq \hat{\mathbf{x}}$. Fix the same item $j$, we have that 
 \begin{align*}
     \sum_{i \in \mathcal{N}} b_{ij}^\ast\log \frac{b_{ij}^\ast}{\tilde{b}_{ij}}&=\sum_{i \in \mathcal{N}} b_{ij}^\ast\log \frac{x_{ij}^\ast p_j^\ast}{\tilde{x}_{ij}\tilde{p}_j}\\
     &=\sum_{i \in \mathcal{N}} b_{ij}^\ast\log \frac{p_j^\ast}{\tilde{p}_j}+\sum_{i \in \mathcal{N}} b_{ij}^\ast\log\frac{x_{ij}^\ast}{\tilde{x}_{ij}}\\
     &=p_j^\ast\log \frac{p_j^\ast}{\tilde{p}_j} +\sum_{i \in \mathcal{N}} p_j^\ast x_{ij}^\ast\log\frac{x_{ij}^\ast}{\tilde{x}_{ij}}\\
     &=p_j^\ast\log \frac{p_j^\ast}{\tilde{p}_j} +p_j^\ast\sum_{i \in \mathcal{N}}  x_{ij}^\ast\log\frac{x_{ij}^\ast}{\tilde{x}_{ij}}
 \end{align*}
 Since a market equilibrium is a fixed point of proportional response, thus we could say
 $$\sum_{i \in \mathcal{N}}\sum_{j \in \mathcal{M}}b_{ij}^\ast \log \frac{b_{ij}^\ast}{\hat{b}_{ij}}=\sum_{i \in \mathcal{N}}\sum_{j \in \mathcal{M}}b_{ij}^\ast\log \frac{b_{ij}^\ast}{\tilde{b}_{ij}},$$
 then
 $$\sum_{j \in \mathcal{M}} p_{j}^\ast \log \frac{p_{j}^\ast}{\hat{p}_{j}}=\sum_{j \in \mathcal{M}} p_j^\ast\log \frac{p_j^\ast}{\tilde{p}_j} +\sum_{j \in \mathcal{M}} p_j^\ast\sum_{i \in \mathcal{N}}  x_{ij}^\ast\log\frac{x_{ij}^\ast}{\tilde{x}_{ij}}.$$
 If $\hat{\mathbf{b}}$ and $\mathbf{b}^\ast$ is the pair that minimize $\sum_{j \in \mathcal{M}} p_j^\ast \log (p_j^\ast/\hat{p_j})$, then 
 $$\sum_{j \in \mathcal{M}} p_j^\ast\log \frac{p_j^\ast}{\tilde{p}_j} \geq \sum_{j \in \mathcal{M}} p_{j}^\ast \log \frac{p_{j}^\ast}{\hat{p}_{j}}=\sum_{j \in \mathcal{M}} p_j^\ast\log \frac{p_j^\ast}{\tilde{p}_j} +\sum_{j \in \mathcal{M}} p_j^\ast\sum_{i \in \mathcal{N}}  x_{ij}^\ast\log\frac{x_{ij}^\ast}{\tilde{x}_{ij}}.$$
 Since the KL divergence is always nonnegative, the term $x_{ij}^\ast\log\frac{x_{ij}^\ast}{\tilde{x}_{ij}}$ must be greater or equal to 0. But the inequality holds if and only if $x_{ij}^\ast\log\frac{x_{ij}^\ast}{\tilde{x}_{ij}}=0$, this means $\tilde{x}=\hat{x}$, and therefore this is contradiction, so $\tilde{\mathbf{b}}=\hat{\mathbf{b}}$. \\
 
 Hence, we conclude that the proportional response leads the bid price to converge to a limit point and the corresponding allocation converges to a single market equilibrium allocation.
\end{proof}
\section{Proofs of Section 5}
\subsection{Proofs of Theorem \ref{nonstationary_independent}} \label{appendix5.1}
\begin{proof}
We start the analysis with the gap in the expected objective value. 

\begin{align*}
        \sum_{t=1}^T\Phi(\mathbf{b}_t)-\sum_{t=1}^T\Phi(\mathbf{b}^\ast)
    &=\underbrace{\sum_{t=1}^T\Phi(\mathbf{b}_t)- \sum_{t=1}^T\Phi(\mathbf{b}^\ast)+\sum_{t=1}^T\varphi(\mathbf{b}^\ast,\varepsilon_t)-\sum_{t=1}^T\varphi(\mathbf{b}_t,\varepsilon_t)}_{\uppercase\expandafter{\romannumeral1}} \\
&+\underbrace{\sum_{t=1}^T\varphi(\mathbf{b}_t,\varepsilon_t)-\sum_{t=1}^T\varphi(\mathbf{b}^\ast,\varepsilon_t)}_{\uppercase\expandafter{\romannumeral2}}.
\end{align*}
Given the history $\mathcal{H}_{t-1}$, let  $\mathcal{P}_t(\cdot \vert \varepsilon_{1:t-1})$ be the  conditional distribution of $\varepsilon_t$, 
we have that 
$$\left\Vert \nabla\Phi(\mathbf{b}_t)-\nabla \varphi(\mathbf{b}_t,\varepsilon_{t}) \right\Vert \leq \Vert \overline{\mathcal{P}} -\mathcal{P}_{t} (\cdot \vert \varepsilon_{1:t-1}) \Vert_{TV}=\delta$$

\textsc{Bounding} \uppercase\expandafter{\romannumeral1}: According to Lemma \ref{lemma3.1}, we can derive that 
$$\Phi(\mathbf{b}_t) -\Phi(\mathbf{b}^\ast) =-\langle \nabla \Phi(\mathbf{b}_t),\mathbf{b}^\ast-\mathbf{b}_t\rangle-D_h(\mathbf{p}^\ast,\mathbf{p}_t),$$
and 
$$\varphi(\mathbf{b}^\ast,\varepsilon_t)-\varphi(\mathbf{b}_t,\varepsilon_t)=\langle \nabla  \varphi(\mathbf{b}_t,\varepsilon_t),\mathbf{b}^\ast-\mathbf{b}_t\rangle+D_h(\mathbf{p}^\ast,\mathbf{p}_t),$$
combining these two equations, we obtain that 
$$\Phi(\mathbf{b}_t) -\Phi(\mathbf{b}^\ast)+\varphi(\mathbf{b}^\ast,\varepsilon_t)-\varphi(\mathbf{b}_t,\varepsilon_t) =\langle\nabla \varphi(\mathbf{b}_t,\varepsilon_t) -\nabla \Phi(\mathbf{b}_t),\mathbf{b}^\ast-\mathbf{b}_t \rangle \leq \Vert  \nabla \varphi(\mathbf{b}_t,\varepsilon_t) -\nabla \Phi(\mathbf{b}_t) \Vert \Vert \mathbf{b}^\ast-\mathbf{b}_t\Vert\leq 2\delta. $$

Moreover, $\Vert \mathbf{b}^\ast-\mathbf{b}_t\Vert \leq 2$.
Then, sum over $t=1,\ldots,T$, we get that $ \uppercase\expandafter{\romannumeral1} \leq  2\delta T$

\textsc{Bounding} \uppercase\expandafter{\romannumeral2}:
Using  Proposition \ref{adverserial_regret}, we have that 
$$\sum_{t=1}^T \varphi(\mathbf{b}_t,\varepsilon_t)-\sum_{t=1}^T \varphi(\mathbf{b}^\ast,\varepsilon_t) \leq \log MN.$$
Therefore, combing the upper bounds of \uppercase\expandafter{\romannumeral1} and \uppercase\expandafter{\romannumeral2}, it follows that 
$$ \sum_{t=1}^T\Phi(\mathbf{b}_t)-\sum_{t=1}^T\Phi(\mathbf{b}^\ast) \leq 2\delta T +\log MN$$
Using the convexity of $\Phi$, we obtain that 
$$T\mathbb{E}\left[\Phi(\mathbf{\hat{b}})-\Phi(\mathbf{b}^\ast) \right]\leq \mathbb{E}\left[\sum_{t=1}^T\Phi(\mathbf{b}_t)-\sum_{t=1}^T\Phi(\mathbf{b}^\ast) \right]\leq 2  \delta T+\log MN,$$
where $\mathbf{\hat{b}} =\frac{1}{T}\sum_{t=1}^T \mathbf{b}_t. $ Consequently, 
$$\mathbb{E}\left[\Phi(\mathbf{\hat{b}})-\Phi(\mathbf{b}^\ast) \right] \leq \frac{2 \delta T+\log MN}{T}.$$
Let $\mathbf{\hat{p}}=\frac{1}{T}\sum_{t=1}^T \mathbf{p}_t$, then $$\hat{p}_j =\frac{1}{T}\sum_{t=1}^T p_{j,t}=\frac{1}{T}\sum_{t=1}^T  \sum_{i\in \mathcal{N}}b_{ij,t}= \sum_{i\in \mathcal{N}} \hat{b}_{ij}, $$
it leads to
$$\mathbb{E}\left[ \Vert\mathbf{\hat{p}}-\mathbf{p}^\ast \Vert^2\right] \leq 2 D_h(\mathbf{\hat{p}},\mathbf{p}^\ast) \leq 2\mathbb{E}\left[\Phi(\mathbf{\hat{b}})-\Phi(\mathbf{b}^\ast) \right] \leq \frac{2 (2\delta T+\log MN)}{T}. $$

\textbf{Individual's regret}: Notice that we measure the regret under the stationary distribution $\overline{\mathcal{P}}$, let $\mathbb{E}[\hat{u}_i]=\frac{1}{T} \sum_{t=1}^T \mathbb{E}[u_{i,t}]$, we can have 
$$ \sum_{t=1}^T \mathbb{E}[ u_i^{\ast}-u_{i,t} ] =T\mathbb{E}\left[ u_i^{\ast}- \frac{1}{T} \sum_{t=1}^Tu_{i,t} \right] = T\mathbb{E}[ u_i^{\ast}- \hat{u}_i ].  $$
Moreover, 
$$\mathbb{E}[\hat{u}_i]=\frac{1}{T} \sum_{t=1}^T\sum_{j \in \mathcal{M}} \mathbb{E}[v_{i} x_{ij,t}] =\frac{1}{T} \sum_{t=1}^T\sum_{j \in \mathcal{M}} \mathbb{E}\left[v_{i} \frac{\hat{b}_{ij}}{\hat{p}_j}\right]=\sum_{{j\in \mathcal{M}}}  \mathbb{E}\left[v_{i} \frac{\hat{b}_{ij}}{\hat{p}_j}\right]. $$
According to Proposition~\ref{generic_upper_bound}, the regret for any buyer $i$ is 
\begin{align*} 
   \sum_{t=1}^T \mathbb{E}[ u_i^{\ast}-u_{i,t} ]&= T\mathbb{E}[ u_i^{\ast}- \hat{u}_i ]\\
   &\leq \frac{\sqrt{2}\overline{v}T}{\underline{p}}\mathbb{E}\left[\sqrt{\Phi({\mathbf{\hat{b}}})-\Phi(\mathbf{b}^\ast)}\right] \\
   &\leq \frac{\sqrt{2}\overline{v}T}{\underline{p}}\sqrt{ \frac{2 \delta T+\log MN}{T}}\\
   &=  \frac{\overline{v}}{\underline{p}} \sqrt{ 4\delta T+2\log MN} \sqrt{T}.
\end{align*}

\end{proof}
\subsection{Proofs of Theorem \ref{nonstationary_Ergodic}} \label{appendix5.2}
\begin{proof}
We start the analysis with the gap in expected objective value. At first, we derive that
\begin{align*}
    &\sum_{t=1}^T \Phi(\mathbf{b}_t)- \sum_{t=1}^T\Phi(\mathbf{b}^\ast)\\ \notag
    =&\underbrace{\sum_{t=1}^{T-\kappa} \Phi(\mathbf{b}_t)-\sum_{t=1}^{T-\kappa} \Phi(\mathbf{b}^\ast)-\sum_{t=1}^{T-\kappa} \varphi(\mathbf{b}_t,\varepsilon_{t+\kappa}) +\sum_{t=1}^{T-\kappa}\varphi(\mathbf{b}^\ast,\varepsilon_{t+\kappa}) }_{\uppercase\expandafter{\romannumeral1}}\\ \notag 
    +&\underbrace{\sum_{t=1}^{T-\kappa}\varphi(\mathbf{b}_{t},\varepsilon_{t+\kappa})-\sum_{t=1}^{T-\kappa}\varphi(\mathbf{b}_{t+\kappa},\varepsilon_{t+\kappa})+\sum_{t=1}^{T} \varphi(\mathbf{b}_t,\varepsilon_t)- \sum_{t=1}^{T} \varphi(\mathbf{b}^\ast,\varepsilon_t)}_{\uppercase\expandafter{\romannumeral2}}\\ \notag
    + &\underbrace{\sum_{t=T-\kappa+1}^T \Phi(\mathbf{b}_t) -\sum_{t=T-\kappa+1}^T \Phi(\mathbf{b}^\ast)-\sum_{t=1}^\kappa \varphi(\mathbf{b}_t,\varepsilon_t) +\sum_{t=1}^\kappa \varphi(\mathbf{b}^\ast,\varepsilon_t) }_{\uppercase\expandafter{\romannumeral3}} 
\end{align*} 
At first, $\varepsilon_{t+\kappa}$  is almost independent of $\mathcal{H}_{t-1}$ and $\nabla \Phi(\mathbf{b})$ is the gradient of the expected value of $\varphi(\mathbf{b},\varepsilon)$ under the stationary distribution, thus
$$\left\Vert \nabla\Phi(\mathbf{b}_{t})-\nabla \varphi(\mathbf{b}_{t},\varepsilon_{t+\kappa}) \right\Vert\leq \delta. $$
\textsc{Bounding} \uppercase\expandafter{\romannumeral1}: 
According to Lemma \ref{lemma3.1}, we can derive that 
$$\Phi(\mathbf{b}_t) -\Phi(\mathbf{b}^\ast) =-\langle \nabla \Phi(\mathbf{b}_t),\mathbf{b}^\ast-\mathbf{b}_t\rangle-D_h(\mathbf{p}^\ast,\mathbf{p}_t),$$
and 
$$\varphi(\mathbf{b}^\ast,\varepsilon_{t+\kappa})-\varphi(\mathbf{b}_t,\varepsilon_{t+\kappa})=\langle \nabla  \varphi(\mathbf{b}_t,\varepsilon_{t+\kappa}),\mathbf{b}^\ast-\mathbf{b}_t\rangle+D_h(\mathbf{p}^\ast,\mathbf{p}_t),$$
combining these two equations, we obtain that 
$$\Phi(\mathbf{b}_t) -\Phi(\mathbf{b}^\ast)+\varphi(\mathbf{b}^\ast,\varepsilon_{t+\kappa})-\varphi(\mathbf{b}_t,\varepsilon_{t+\kappa}) =\langle\nabla \varphi(\mathbf{b}_t,\varepsilon_{t+\kappa}) -\nabla \Phi(\mathbf{b}_t),\mathbf{b}^\ast-\mathbf{b}_t \rangle.
$$ 
Summing over $t=1,\ldots,T-\kappa$, we obtain that 
$$\uppercase\expandafter{\romannumeral1} = \sum_{t=1}^{T-\kappa} \langle\nabla \varphi(\mathbf{b}_t,\varepsilon_{t+\kappa}) -\nabla \Phi(\mathbf{b}_t),\mathbf{b}^\ast-\mathbf{b}_t \rangle. $$
\textsc{Bounding} \uppercase\expandafter{\romannumeral2}:
Using  Proposition \ref{adverserial_regret}, we have that 
$$\sum_{t=1}^T \varphi(\mathbf{b}_t,\varepsilon_t)-\sum_{t=1}^T \varphi(\mathbf{b}^\ast,\varepsilon_t) \leq \log MN.$$ 
Using Lemma \ref{lemma3.1} again, we have that 
$$\varphi(\mathbf{b}_t,\varepsilon_{t+\kappa})-\varphi(\mathbf{b}_{t+\kappa},\varepsilon_{t+\kappa})= \langle\nabla \varphi(\mathbf{b}_{t},\varepsilon_{t+\kappa}) , \mathbf{b}_t- \mathbf{b}_{t+\kappa}\rangle -D_h (\mathbf{p}_{t+\kappa},\mathbf{p}_{t}).$$
Together with the upper bound of \uppercase\expandafter{\romannumeral1}, it follows that 
\begin{align*}
    &\uppercase\expandafter{\romannumeral1} +\uppercase\expandafter{\romannumeral2} \\
    & \leq \sum_{t=1}^{T-\kappa} \langle\nabla \varphi(\mathbf{b}_t,\varepsilon_{t+\kappa}) -\nabla \Phi(\mathbf{b}_t),\mathbf{b}^\ast-\mathbf{b}_t \rangle + \sum_{t=1}^{T-\kappa} \langle\nabla \varphi(\mathbf{b}_{t},\varepsilon_{t+\kappa}), \mathbf{b}_t- \mathbf{b}_{t+\kappa}\rangle- \sum_{t=1}^{T-\kappa}D_h (\mathbf{p}_{t+\kappa},\mathbf{p}_{t}) +\log MN \\
    &= \sum_{t=1}^{T-\kappa} \langle\nabla \varphi(\mathbf{b}_t,\varepsilon_{t+\kappa}) -\nabla \Phi(\mathbf{b}_t),\mathbf{b}^\ast-\mathbf{b}_t \rangle + \sum_{t=1}^{T-\kappa} \langle\nabla \varphi(\mathbf{b}_{t},\varepsilon_{t+\kappa}), \mathbf{b}_t- \mathbf{b}_{t+\kappa}\rangle \\
    &-\sum_{t=1}^{T-\kappa}  \langle\nabla \Phi(\mathbf{b}_t),\mathbf{b}^\ast-\mathbf{b}_{t+\kappa} \rangle +\sum_{t=1}^{T-\kappa}  \langle\nabla \Phi(\mathbf{b}_t),\mathbf{b}^\ast-\mathbf{b}_{t+\kappa} \rangle-\sum_{t=1}^{T-\kappa}D_h (\mathbf{p}_{t+\kappa},\mathbf{p}_{t})+\log MN\\
    &=\sum_{t=1}^{T-\kappa} \langle\nabla \varphi(\mathbf{b}_t,\varepsilon_{t+\kappa}) -\nabla \Phi(\mathbf{b}_t),\mathbf{b}^\ast-\mathbf{b}_{t+\kappa} \rangle + \sum_{t=1}^{T-\kappa}\langle\nabla \Phi(\mathbf{b}_t),\mathbf{b}_t-\mathbf{b}_{t+\kappa} \rangle-\sum_{t=1}^{T-\kappa}D_h (\mathbf{p}_{t+\kappa},\mathbf{p}_{t})+\log MN\\
    & \leq \sum_{t=1}^{T-\kappa} \Vert \nabla \varphi(\mathbf{b}_t,\varepsilon_{t+\kappa}) -\nabla \Phi(\mathbf{b}_t) \Vert \Vert\mathbf{b}^\ast-\mathbf{b}_{t+\kappa} \Vert + \sum_{t=1}^{T-\kappa} \Phi(\mathbf{b}_t) -\sum_{t=1}^{T-\kappa} \Phi  (\mathbf{b}_{t+\kappa}) +\log MN \\
    & \leq 2(T-\kappa)\delta + \sum_{t=1}^\kappa \Phi(\mathbf{b}_t) -\sum_{t=T-\kappa+1}^T  \Phi(\mathbf{b}_t) +\log MN,
\end{align*}
where the second inequality follows the Cauchy-Schwarz inequality in $\ell_1$-norm (i.e., $\vert \mathbf{x},\mathbf{y}\vert \leq \Vert \mathbf{x} \Vert \Vert \mathbf{y} \Vert$); note that the diameter of an unit simplex in $\ell_1$-norm is 2 such that $\Vert \mathbf{b}^\ast- \mathbf{b}_{t+\kappa}\Vert \leq 2$ and $\Phi(\mathbf{b}_{t})+\langle \nabla\Phi (\mathbf{b}_t), \mathbf{b}_{t+\kappa}-\mathbf{b}_t\rangle+D_h(\mathbf{p}_{t+\kappa},\mathbf{p}_t) =\Phi(\mathbf{b}_{t+\kappa})$, we obtain the last inequality. 
Add $\uppercase\expandafter{\romannumeral1}$, $\uppercase\expandafter{\romannumeral 2}$ and $\uppercase\expandafter{\romannumeral 3}$, we get 
\begin{align*}
&\uppercase\expandafter{\romannumeral1}+\uppercase\expandafter{\romannumeral2} +\uppercase\expandafter{\romannumeral3} \\
    & \leq \sum_{t=T-\kappa+1}^T \Phi(\mathbf{b}_t) -\sum_{t=T-\kappa+1}^T \Phi(\mathbf{b}^\ast)-\sum_{t=1}^\kappa \varphi(\mathbf{b}_t,\varepsilon_t) +\sum_{t=1}^\kappa \varphi(\mathbf{b}^\ast,\varepsilon_t) \\
    &+ \sum_{t=1}^\kappa \Phi(\mathbf{b}_t) -\sum_{t=T-\kappa+1}^T  \Phi(\mathbf{b}_t) +2(T-\kappa)\delta +\log MN \\
    &= \sum_{t=1}^\kappa \Phi(\mathbf{b}_t) -\sum_{t=T-\kappa+1}^T \Phi(\mathbf{b}^\ast)-\sum_{t=1}^\kappa \varphi(\mathbf{b}_t,\varepsilon_t) +\sum_{t=1}^\kappa \varphi(\mathbf{b}^\ast,\varepsilon_t)+2(T-\kappa)\delta+ \log MN \\
    &=\sum_{t=1}^\kappa \Phi(\mathbf{b}_t) -\sum_{t=1}^\kappa \Phi(\mathbf{b}^\ast)-\sum_{t=1}^\kappa \varphi(\mathbf{b}_t,\varepsilon_t) +\sum_{t=1}^\kappa \varphi(\mathbf{b}^\ast,\varepsilon_t)+2(T-\kappa)\delta+ \log MN \\
    &\leq 2(T-\kappa)\delta+ 2\kappa \overline{\varphi} +\log MN, 
\end{align*}
where the second equality comes from the fact that $\sum_{t=T-\kappa+1}^T \Phi(\mathbf{b}^\ast)=\sum_{t=1}^\kappa \Phi(\mathbf{b}^\ast)$, the last inequality uses the assumption of $\overline{\varphi}$. 
Therefore, we obtain 
$$\sum_{t=1}^T\Phi(\mathbf{b}_t)-\sum_{t=1}^T\Phi(\mathbf{b}^\ast) \leq 2(T-\kappa)\delta+2\kappa \overline{\varphi}+ \log MN.$$
Using the convexity of $\Phi$, we obtain that 
$$T\mathbb{E}\left[\Phi(\mathbf{\hat{b}})-\Phi(\mathbf{b}^\ast) \right]\leq \mathbb{E}\left[\sum_{t=1}^T\Phi(\mathbf{b}_t)-\sum_{t=1}^T\Phi(\mathbf{b}^\ast) \right]\leq 2(T-\kappa)\delta+2\kappa \overline{\varphi}+ \log MN,$$
Let $\mathbf{\hat{p}}=\frac{1}{T}\sum_{t=1}^T \mathbf{p}_t$, then $$\hat{p}_j =\frac{1}{T}\sum_{t=1}^T p_{j,t}=\frac{1}{T}\sum_{t=1}^T \sum_{i\in \mathcal{N}}  b_{ij,t}= \sum_{i\in \mathcal{N}} \hat{b}_{ij}, $$
it leads to
$$\mathbb{E}\left[ \Vert\mathbf{\hat{p}}-\mathbf{p}^\ast \Vert^2\right] \leq 2 D_h(\mathbf{\hat{p}},\mathbf{p}^\ast) \leq 2\mathbb{E}\left[\Phi(\mathbf{\hat{b}})-\Phi(\mathbf{b}^\ast) \right] \leq \frac{2 (2(T-\kappa)\delta+2\kappa \overline{\varphi}+ \log MN)}{T}. $$
\textbf{Individual's regret}:
According to Proposition~\ref{generic_upper_bound},
\begin{align*} \label{regret_nonstationary_Ergodic}
   \sum_{t=1}^T \mathbb{E}[ u_i^{\ast}-u_{i,t} ]&= T\mathbb{E}[ u_i^{\ast}- \hat{u}_i ]\\
   &\leq \frac{\sqrt{2}\overline{v}T}{\underline{p}}\mathbb{E}\left[\sqrt{\Phi({\mathbf{\hat{b}}})-\Phi(\mathbf{b}^\ast)}\right] \\
   &\leq \frac{\sqrt{2}\overline{v}T}{\underline{p}}\sqrt{ \frac{2(T-\kappa)\delta+2\kappa \overline{\varphi}+ \log MN}{T}}\\
   &=  \frac{\overline{v}}{\underline{p}} \sqrt{4(T-\kappa)\delta+4\kappa \overline{\varphi}+ 2\log MN} \sqrt{T}.
\end{align*}

\end{proof}
\subsection{Proofs of Theorem \ref{nonstationary_Periodic}} \label{appendix5.3}
\begin{proof}
Fix a partition $q$, we can derive that 
\begin{align*}
    &\sum_{t=t_q}^{t_{q+1}-1}\Phi(\mathbf{b}_t) -\sum_{t=t_q}^{t_{q+1}-1}\Phi(\mathbf{b}^\ast)\\
    &=\underbrace{\sum_{t=t_q}^{t_{q+1}-1}\Phi(\mathbf{b}_t)- \sum_{t=t_q}^{t_{q+1}-1}\Phi(\mathbf{b}_{t_q})+\sum_{t=t_q}^{t_{q+1}-1}\varphi(\mathbf{b}_{t_q},\varepsilon_t)-\sum_{t=t_q}^{t_{q+1}-1}\varphi(\mathbf{b}_{t},\varepsilon_t)}_{\uppercase\expandafter{\romannumeral1}} \\
    &+\underbrace{\sum_{t=t_q}^{t_{q+1}-1}\Phi(\mathbf{b}_{t_q})- \sum_{t=t_q}^{t_{q+1}-1}\varphi(\mathbf{b}_{t_q},\varepsilon_t) +\sum_{t=t_q}^{t_{q+1}-1} \varphi (\mathbf{b}^\ast,\varepsilon_t)-\sum_{t=t_q}^{t_{q+1}-1}\Phi(\mathbf{b}^\ast)}_{\uppercase\expandafter{\romannumeral2}} \\
    &+\underbrace{\sum_{t=t_q}^{t_{q+1}-1}\varphi(\mathbf{b}_{t},\varepsilon_t)-\sum_{t=t_q}^{t_{q+1}-1}\varphi(\mathbf{b}^\ast,\varepsilon_t)}_{\uppercase\expandafter{\romannumeral3}}
\end{align*}
The noises $\varepsilon_t$ within
the same partition can be arbitrarily correlated but partitions are independently and identically distributed. 
Given the history $\mathcal{H}_{t_q-1}$,  the bid  $\mathbf{b}_{t_q}$ is independent of $\mathcal{H}_{t_q-1}$ because the joint distribution of each partition is independent with each other, thus, 
$$\vert  \nabla\Phi(\mathbf{b}_{t_q})-\nabla \varphi(\mathbf{b}_{t_q},\varepsilon_{t}) \vert \leq \Vert \overline{\mathcal{P}} -\mathcal{P}_{t}  \Vert_{TV}=\delta.$$
Bounding \uppercase\expandafter{\romannumeral1}+\uppercase\expandafter{\romannumeral2}: According to Lemma \ref{lemma3.1}, we have that, 
$$\Phi(\mathbf{b}_t)-\Phi(\mathbf{b}_{t_q})=\langle \nabla \Phi(\mathbf{b}_{t_q}),\mathbf{b}_t-\mathbf{b}_{t_q}\rangle+D_h(\mathbf{p}_t,\mathbf{p}_{t_q}) $$
Similarly, 
$$\varphi(\mathbf{b}_{t_q},\varepsilon_t)-\varphi(\mathbf{b}_{t},\varepsilon_t)=\langle \nabla \varphi(\mathbf{b}_{t_q},\varepsilon_t),\mathbf{b}_{t_q}-\mathbf{b}_t\rangle-D_h(\mathbf{p}_{t},\mathbf{p}_{t_q}).$$
Thus, 
$$ \Phi(\mathbf{b}_t)-\Phi(\mathbf{b}_{t_q})+\varphi(\mathbf{b}_{t_q},\varepsilon_t)-\varphi(\mathbf{b}_{t},\varepsilon_t)=\langle \nabla \Phi(\mathbf{b}_{t_q})-\nabla \varphi(\mathbf{b}_{t_q},\varepsilon_t),\mathbf{b}_t-\mathbf{b}_{t_q}\rangle.$$
 We also have that 
\begin{align*}
   &\Phi(\mathbf{b}_{t_q})-\Phi(\mathbf{b}^\ast)+\varphi(\mathbf{b}^\ast,\varepsilon_t)-\varphi(\mathbf{b}_{t_q},\varepsilon_t) \\
   &=\langle \nabla \Phi(\mathbf{b}_{t_q}),\mathbf{b}_{t_q}-\mathbf{b}^\ast \rangle- D_h(\mathbf{b}^\ast, \mathbf{b}_{t_q}) +\langle \nabla \varphi(\mathbf{b}_{t_q},\varepsilon_t), \mathbf{b}^\ast-\mathbf{b}_{t_q} \rangle+D_h(\mathbf{b}^\ast, \mathbf{b}_{t_q}) \\
   & = \langle \nabla \Phi(\mathbf{b}_{t_q})-\nabla \varphi(\mathbf{b}_{t_q},\varepsilon_t),\mathbf{b}_{t_q}-\mathbf{b}^\ast \rangle. 
\end{align*}
Thus, we obtain that 
\begin{align*}
    &\Phi(\mathbf{b}_t)-\Phi(\mathbf{b}_{t_q})+\varphi(\mathbf{b}_{t_q},\varepsilon_t)-\varphi(\mathbf{b}_{t},\varepsilon_t)+\Phi(\mathbf{b}_{t_q})-\Phi(\mathbf{b}^\ast)+\varphi(\mathbf{b}^\ast,\varepsilon_t)-\varphi(\mathbf{b}_{t_q},\varepsilon_t) \\
   &= \langle \nabla \Phi(\mathbf{b}_{t_q})-\nabla \varphi(\mathbf{b}_{t_q},\varepsilon_t),\mathbf{b}_t-\mathbf{b}_{t_q}\rangle +\langle \nabla \Phi(\mathbf{b}_{t_q})-\nabla \varphi(\mathbf{b}_{t_q},\varepsilon_t),\mathbf{b}_{t_q}-\mathbf{b}^\ast \rangle \\
   &= \langle \nabla \Phi(\mathbf{b}_{t_q})-\nabla \varphi(\mathbf{b}_{t_q},\varepsilon_t),\mathbf{b}_t -\mathbf{b}^\ast \rangle \\
   & \leq \Vert \nabla  \Phi(\mathbf{b}_{t_q})-\nabla \varphi(\mathbf{b}_{t_q},\varepsilon_t) \Vert \Vert \mathbf{b}_t -\mathbf{b}^\ast \Vert \\
   & \leq 2\delta,
\end{align*}
where the first inequality follows the Cauchy-Schwarz inequality in $\ell_1$ norm (i.e., $\vert \mathbf{x},\mathbf{y}\vert \leq \Vert \mathbf{x} \Vert \Vert \mathbf{y} \Vert$) and the  fact that the diameter of an unit simplex in $\ell_1$ norm is 2 such that $\Vert \mathbf{b}_t- \mathbf{b}^\ast\Vert \leq 2$.  Summing over $t=t_q,\ldots,t_{q+1}-1$, we have that 
$$\uppercase\expandafter{\romannumeral1} +\uppercase\expandafter{\romannumeral2} \leq 2I_q\delta.$$
Therefore, we obtain that 
 $\mathbb{E}[\uppercase\expandafter{\romannumeral1}+\uppercase\expandafter{\romannumeral2}] \leq 2I_q \delta$. \\

Combing the bounds of \uppercase\expandafter{\romannumeral1} , \uppercase\expandafter{\romannumeral2} and  \uppercase\expandafter{\romannumeral3}, we obtain
$$\sum_{t=1}^T\Phi(\mathbf{b}_t) -\sum_{t=1}^T\Phi(\mathbf{b}^\ast) = 2\delta T +\log MN, $$
where the upper bound of \uppercase\expandafter{\romannumeral3} follows 
the result in  Proposition \ref{adverserial_regret}.  
Using the convexity of $\Phi$, we obtain that 
$$T\mathbb{E}\left[\Phi(\mathbf{\hat{b}})-\Phi(\mathbf{b}^\ast) \right]\leq \mathbb{E}\left[\sum_{t=1}^T\Phi(\mathbf{b}_t)-\sum_{t=1}^T\Phi(\mathbf{b}^\ast) \right]\leq \log MN+ 2\delta T.$$
Let $\mathbf{\hat{p}}=\frac{1}{T}\sum_{t=1}^T \mathbf{p}_t$, then $$\hat{p}_j =\frac{1}{T}\sum_{t=1}^T p_{j,t}=\frac{1}{T}\sum_{t=1}^T\sum_{i\in \mathcal{N}}  b_{ij,t}= \sum_{i\in \mathcal{N}}  \hat{b}_{ij}, $$
it leads to
$$\mathbb{E}\left[ \Vert\mathbf{\hat{p}}-\mathbf{p}^\ast \Vert^2\right] \leq 2 D_h(\mathbf{\hat{p}},\mathbf{p}^\ast) \leq 2\mathbb{E}\left[\Phi(\mathbf{\hat{b}})-\Phi(\mathbf{b}^\ast) \right] \leq \frac{2 \log MN+ 4\delta T}{T}. $$
\textbf{Individual's regret}:
According to Proposition~\ref{generic_upper_bound},
\begin{align*} 
   \sum_{t=1}^T \mathbb{E}[ u_i^{\ast}-u_{i,t} ]&= T\mathbb{E}[ u_i^{\ast}- \hat{u}_i ]\\
   &\leq \frac{\sqrt{2}\overline{v}T}{\underline{p}}\mathbb{E}\left[\sqrt{\Phi({\mathbf{\hat{b}}})-\Phi(\mathbf{b}^\ast)}\right] \\
   &\leq \frac{\sqrt{2}\overline{v}T}{\underline{p}}\sqrt{ \frac{2\delta T+\log MN}{T}}\\
   &=  \frac{\overline{v}}{\underline{p}} \sqrt{4\delta T+2\log MN} \sqrt{T}.
\end{align*}

\end{proof}
\section{Supplementary Lemmas}
Here we restate the update rule (\ref{OMDupdate}) for ease of reading
$$\mathbf{b}_{i,t}=\argmin_{\sum_{j \in \mathcal{M}} b_{ij,t}=B_i} \left\lbrace  \langle \nabla \varphi(\mathbf{b}_{t-1},\varepsilon_t),\mathbf{b}_{i} {-\mathbf{b}_{i,t-1}}\rangle+\sum_{j \in \mathcal{M}} b_{ij}\log \left(\frac{b_{ij}}{b_{ij,t-1}}\right)\right\rbrace. $$
\begin{lemma}\label{pythagorean}
    For any $t$, if $\mathbf{b}_t$ is a minimum of (\ref{OMDupdate}), then for any $\mathbf{b}^\ast$
    $$\langle \nabla \varphi(\mathbf{b}_{t-1},\varepsilon_t), \mathbf{b}_t-\mathbf{b}_{t-1}\rangle+D_h(\mathbf{b}_t,\mathbf{b}_{t-1})\leq \langle \nabla \varphi(\mathbf{b}_{t-1},\varepsilon_t), \mathbf{b}^\ast-\mathbf{b}_{t-1}\rangle+D_h(\mathbf{b}^\ast,\mathbf{b}_{t-1})-D_h(\mathbf{b}^\ast,\mathbf{b}_{t}). $$
\end{lemma}
\proof{This lemma is a quite standard property for mirror descent, see \cite{chen1993convergence, beck2003mirror}, etc. We present the proofs here for the readers to better understand this property in an online mirror descent version.
Since $\mathbf{b}_t$ is the minimum of (\ref{OMDupdate}), using the optimality condition,
$$\langle \nabla \varphi(\mathbf{b}_{t-1},\varepsilon_t)+\nabla h(\mathbf{b}_t)- \nabla h(\mathbf{b}_{t-1}),\mathbf{b}^\ast-\mathbf{b}_t\rangle \geq 0.$$
Given the ``three-point identity" for the Bregman Divergences, 
$$D_h(c,a)+D_h(a,b)-D_h(c,b)=\langle\nabla h(b)-\nabla h(a),c-a\rangle,$$
we replace  $a=\mathbf{b}_t$, $b=\mathbf{b}_{t-1}$, $c=\mathbf{b}^\ast$ and obtain that 
\begin{equation}\label{threepoint}
    \langle \nabla \varphi(\mathbf{b}_{t-1},\varepsilon_t),\mathbf{b}^\ast-\mathbf{b}_t \rangle \geq  D_h(\mathbf{b}_t,\mathbf{b}_{t-1})+D_h(\mathbf{b}^\ast,\mathbf{b}_{t})-D_h(\mathbf{b}^\ast,\mathbf{b}_{t-1}).
\end{equation}
Rearranging (\ref{threepoint}) leads to
$$\langle  \nabla \varphi(\mathbf{b}_{t-1},\varepsilon_t),\mathbf{b}_t\rangle + D_h(\mathbf{b}_t,\mathbf{b}_{t-1}) \leq  \langle  \nabla \varphi(\mathbf{b}_{t-1},\varepsilon_t),\mathbf{b}^\ast \rangle + D_h(\mathbf{b}^\ast,\mathbf{b}_{t-1})- D_h(\mathbf{b}^\ast,\mathbf{b}_{t}).$$
By subtracting $ \langle  \nabla \varphi(\mathbf{b}_{t-1},\varepsilon_t),\mathbf{b}_{t-1} \rangle $ on both sides, we obtain that 
 $$\langle \nabla \varphi(\mathbf{b}_{t-1},\varepsilon_t), \mathbf{b}_t-\mathbf{b}_{t-1}\rangle+D_h(\mathbf{b}_t,\mathbf{b}_{t-1})\leq \langle \nabla \varphi(\mathbf{b}_{t-1},\varepsilon_t), \mathbf{b}^\ast-\mathbf{b}_{t-1}\rangle+D_h(\mathbf{b}^\ast,\mathbf{b}_{t-1})-D_h(\mathbf{b}^\ast,\mathbf{b}_{t}). $$
}
\begin{proposition} \label{adverserial_regret}
If every buyer $i \in \mathcal{N}$ uses (\ref{OMDupdate}) to update bids, then   for any feasible $\mathbf{b}^\ast$,  we have
\begin{equation}\label{objective_diff}
    \varphi(\mathbf{b}_t,\varepsilon_t)-\varphi(\mathbf{b}^\ast,\varepsilon_t) \leq  D_h(\mathbf{b}^\ast,\mathbf{b}_{t-1})-D_h(\mathbf{b}^\ast,\mathbf{b}_{t}).
\end{equation}
Summing over $t=1,\ldots,T$, we obtain that 
\begin{equation}\label{staticregret}
    \sum_{t=1}^T \varphi(\mathbf{b}_t, \varepsilon_t)- \sum_{t=1}^T\varphi(\mathbf{b}^\ast,\varepsilon_t) \leq  D_h(\mathbf{b}^\ast,\mathbf{b}_{1}).
\end{equation}
\end{proposition}
\begin{proof}
Using Lemma \ref{pythagorean}, it holds that
$$\langle \nabla \varphi(\mathbf{b}_{t-1},\varepsilon_t), \mathbf{b}_t-\mathbf{b}_{t-1}\rangle+D_h(\mathbf{b}_t,\mathbf{b}_{t-1})\leq \langle \nabla \varphi(\mathbf{b}_{t-1},\varepsilon_t), \mathbf{b}^\ast-\mathbf{b}_{t-1}\rangle+D_h(\mathbf{b}^\ast,\mathbf{b}_{t-1})-D_h(\mathbf{b}^\ast,\mathbf{b}_{t}). $$
Adding $\varphi(\mathbf{b}_{t-1},\varepsilon_t)$ on both sides, we obtain that 
$$\ell_{\varphi}(\mathbf{b}_t,\mathbf{b}_{t-1},\varepsilon_t)+D_h(\mathbf{b}_t, \mathbf{b}_{t-1})  \leq \ell_{\varphi}(\mathbf{b}^\ast,\mathbf{b}_{t-1},\varepsilon_t)+D_h(\mathbf{b}^\ast,\mathbf{b}_{t-1})-D_h(\mathbf{b}^\ast, \mathbf{b}_{t}).$$
By the convexity of $\varphi(\cdot)$, we have that $\varphi(\mathbf{b}^\ast,\varepsilon_t) \geq \varphi(\mathbf{b}_{t-1},\varepsilon_t)+\langle  \nabla \varphi(\mathbf{b}_{t-1},\varepsilon_t),\mathbf{b}^\ast - \mathbf{b}_{t-1}, \rangle = \ell_{\varphi}(\mathbf{b}^\ast,\mathbf{b}_{t-1},\varepsilon_t)$. Using the relative smoothness condition, we have that 
\begin{align} \label{obj_diff_proof}
    \varphi(\mathbf{b}_t,\varepsilon_t) &\leq \ell_{\varphi}(\mathbf{b}_t,\mathbf{b}_{t-1},\varepsilon_t)+D_h(\mathbf{b}_t, \mathbf{b}_{t-1}) \notag \\
    &\leq \ell_{\varphi}(\mathbf{b}^\ast,\mathbf{b}_{t-1},\varepsilon_t)+D_h(\mathbf{b}^\ast,\mathbf{b}_{t-1})-D_h(\mathbf{b}^\ast, \mathbf{b}_{t}) \notag \\
    &\leq \varphi(\mathbf{b}^\ast,\varepsilon_t)+D_h(\mathbf{b}^\ast,\mathbf{b}_{t-1})-D_h(\mathbf{b}^\ast, \mathbf{b}_{t}).
\end{align}
\end{proof}

\begin{lemma} If $b_{ij,1}= \frac{B_i}{M}$, then $D_h(\mathbf{b}^\ast,\mathbf{b}_1) \leq \log MN$. \label{D(b,b1)}
\end{lemma}
\proof{
\begin{align*}
D_h(\mathbf{b}^\ast,\mathbf{b}_1)&=\sum_{i \in \mathcal{N}}\sum_{j \in \mathcal{M}}b_{ij}^\ast \log \frac{b_{ij}^\ast}{b_{ij,1}}=\sum_{i \in \mathcal{N}}\sum_{j \in \mathcal{M}}b_{ij}^\ast \log \frac{Mb_{ij}^\ast}{B_i}\\
    &=\sum_{i \in \mathcal{N}}\sum_{j \in \mathcal{M}}b_{ij}^\ast\log M +\sum_{i \in \mathcal{N}}\sum_{j \in \mathcal{M}}b_{ij}^\ast \log b_{ij}^\ast-\sum_{i \in \mathcal{N}}\sum_{j \in \mathcal{M}}b_{ij}^\ast\log B_i\\
    &\leq \log M+\log N.
\end{align*}
The last inequality comes from the fact that 
$\sum_{i \in \mathcal{N}}\sum_{j \in \mathcal{M}} b_{ij}^\ast =1$, $\log b_{ij}^\ast \leq 0$ and $-\sum_{i \in \mathcal{N}}\sum_{j \in \mathcal{M}}b_{ij}^\ast\log B_i=-\sum_{i \in \mathcal{N}} B_i \log B_i \leq \log N$.
This lemma follows Lemma 13 in \cite{birnbaum2011distributed}.
}
\begin{lemma} For any $t=1,\ldots,T$, $\mathbf{b}_t$ is a minimum of (\ref{OMDupdate}), we have that 
$$\mathbf{b}_{t}:= \argmin \langle \nabla \varphi(\mathbf{b}_{t-1},\varepsilon_t),\mathbf{b}_t-\mathbf{b}_{t-1} \rangle]+D_h(\mathbf{b}_t,\mathbf{b}_{t-1}),$$
and 
$$ \varphi(\mathbf{b}_t,\varepsilon_t) \leq \varphi(\mathbf{b}_{t-1},\varepsilon_t).$$
\end{lemma}
\proof{Since $\mathbf{b}_t$ is a minimum of (\ref{OMDupdate}), then
\begin{align*}
    \mathbf{b}_{t}&:= \argmin \left\lbrace\ell_\varphi(\mathbf{b}_t,\mathbf{b}_{t-1},\varepsilon_t) + D_h(\mathbf{b}_t,\mathbf{b}_{t-1}) \right\rbrace\\ \notag
  &:=\argmin \left\lbrace \varphi(\mathbf{b}_{t-1},\varepsilon_t)+\langle \nabla \varphi(\mathbf{b}_{t-1},\varepsilon_t),\mathbf{b}_t-\mathbf{b}_{t-1} \rangle+D_h(\mathbf{b}_t,\mathbf{b}_{t-1}) \right\rbrace \\ \notag
  &:= \argmin \left\lbrace \langle \nabla \varphi(\mathbf{b}_{t-1},\varepsilon_t),\mathbf{b}_t-\mathbf{b}_{t-1} \rangle]+D_h(\mathbf{b}_t,\mathbf{b}_{t-1})\right\rbrace.
\end{align*}
The last equality holds due to $\varphi(\mathbf{b}_{t-1},\varepsilon_t)$ being independent of $\mathbf{b}_t$. And we have that 
\begin{align*}
    \varphi(\mathbf{b}_t,\varepsilon_t) & \leq  \ell_\varphi(\mathbf{b}_t,\mathbf{b}_{t-1},\varepsilon_t) + D_h(\mathbf{b}_t,\mathbf{b}_{t-1}) \\
    & \leq \ell_\varphi(\mathbf{b}_{t-1},\mathbf{b}_{t-1},\varepsilon_t) + D_h(\mathbf{b}_{t-1},\mathbf{b}_{t-1}) \\
    &=\varphi(\mathbf{b}_{t-1},\varepsilon_t).
\end{align*}
The first inequality comes from the relative smoothness condition. The second inequality comes from that $\mathbf{b}_t$ is a minimum of (\ref{OMDupdate}).}
\end{document}